\documentclass[conference]{IEEEtran}
\IEEEoverridecommandlockouts

\usepackage{cite}
\usepackage{amsmath,amssymb,amsfonts}
\usepackage{amsthm}
\usepackage{algorithmic}
\usepackage[ruled,vlined]{algorithm2e}
\usepackage{graphicx}
\usepackage{textcomp}
\usepackage{xcolor}
\usepackage{xspace}
\usepackage {tikz}
\usetikzlibrary {positioning}
\usepackage{pgfplotstable}
\usepackage{pgfplots}
\usepackage{wrapfig}
\usepackage{hyperref}
\usepackage{array}

\usepackage{balance}
\usepackage{enumitem}

\usepackage{graphicx}

\newcommand{\att}{attitude\xspace}
\newcommand{\attup}{Attitude\xspace}
\newcommand{\attic}{AIC\xspace}
\newcommand{\attof}{\mathtt{Att}}
\newcommand{\attin}{\mathtt{AttIn}}
\newcommand{\sigatt}{\sigma_{\attof}}
\newcommand{\sigattinf}{\sigma_{\mathtt{Act}}}
\newcommand{\OPT}{\sigatt(S^*)}
\newcommand{\expectedDeg}{\expectation_{g\sim G}[deg(v)]}

\newcommand{\Eff}{Actionable Attitude\xspace}
\newcommand{\EffUpperCase}{ACTIONABLE ATTITUDE\xspace}
\newcommand{\ExpEff}{\sigma_{\mathtt{Act}}}

\newcommand{\expectation}{\mathbb{E}}

\newtheorem{definition}{\textsc{Definition}}
\newtheorem{theorem}{\textsc{Theorem}}
\newtheorem{lemma}{\textsc{Lemma}}

\newtheorem{problem}{\textsc{Problem}}

\makeatletter
\newcommand{\linebreakand}{%
  \end{@IEEEauthorhalign}
  \hfill\mbox{}\par
  \mbox{}\hfill\begin{@IEEEauthorhalign}
}
\makeatother

\begin{document}


\title{Measuring the Impact of Influence on Individuals: Roadmap
to Quantifying Attitude\thanks{Fu and Padmanabhan  are equal contributors. Research supported in part by NSF grants  1849053, 1934884}}

\author{
  \IEEEauthorblockN{Xiaoyun Fu}
  \IEEEauthorblockA{
    Department of Computer Science\\
    Iowa State University \\
    xfu@iastate.edu
  }
  \and
  \IEEEauthorblockN{Madhavan Padmanabhan}
  \IEEEauthorblockA{
    Department of Computer Science\\
    Iowa State University \\
    madhavrp@iastate.edu
  }
  \and
   \IEEEauthorblockN{Raj Gaurav Kumar}
  \IEEEauthorblockA{
    Department of Computer Science\\
    Iowa State University \\
    gaurav@iastate.edu
  }
  \linebreakand
  
  \IEEEauthorblockN{Samik Basu}
  \IEEEauthorblockA{
    Department of Computer Science\\
    Iowa State University \\
    sbasu@iastate.edu
  }
    \and
  \IEEEauthorblockN{Shawn Dorius}
  \IEEEauthorblockA{
    Department of Sociology\\
    Iowa State University \\
    sdorius@iastate.edu
  }
  \and
  \IEEEauthorblockN{A. Pavan}
  \IEEEauthorblockA{
    Department of  Computer Science\\
    Iowa State University \\
    pavan@iastate.edu
  }
}

\maketitle

\setlength{\textfloatsep}{5pt}
\begin{abstract}
Influence diffusion has been central to the study of the propagation of
information in social networks, where influence is typically modeled as
a binary property of entities: influenced or not influenced. We
introduce the notion of attitude, which, as described in social
psychology, is the degree by which an entity is influenced by the
information. We present an information diffusion model that quantifies
the degree of influence, i.e., attitude of individuals, in a social
network.  With this model, we formulate and study attitude
maximization problem.  We prove that the function for computing
attitude is monotonic and sub-modular, and the attitude maximization
problem is NP-Hard. We present a greedy algorithm for maximization
with an approximation guarantee of $(1-1/e)$.  Using the same model,
we also introduce the notion of ``actionable'' attitude with the aim
to study the scenarios where attaining individuals with high attitude
is objectively more important than maximizing the attitude of the
entire network. We show that the function for computing actionable
attitude, unlike that for computing attitude, is non-submodular and
however is \emph{approximately submodular}. We present approximation
algorithm for maximizing actionable attitude in a network. We
experimentally evaluated our algorithms and study empirical properties
of the attitude of nodes in network such as spatial and value
distribution of high attitude nodes. 
\end{abstract}


\section{Introduction}
\label{sec:intro}

The proliferation of social networks and their influence in modern
society led to a large body of research in several scientific domains
that focus on utilizing and explaining the significance of the impact
of social networks. One of the key problems investigated is to
understand the diffusion of information/influence propagation in
social networks.  Diffusion refers to the (probabilistic) behavior of
the interaction between the entities in the network describing
when/how an entity is influenced by the actions of its neighbors.

Seminal works of Domingos and Richardson, and Kempe {\it et al.}
proposed two popular models for information diffusion/influence
propagation---Independent Cascade and Linear
Threshold~\cite{domingos:kdd01,kempe:kdd03}. In these models, a node
of a network is said to be influenced if it receives the information
originated at the seed set.  This concept of influence is binary: an
entity is either influenced or is not influenced.  Real-world
experience shows that not all influenced individuals are the
same. I.e, some individuals are more \emph{strongly} influenced by
certain information compared to others. Thus, the {\em strength of
  influence} can vary from one individual to the other.  This
phenomenon has been pointed out in social sciences literature.
 
Within social psychology, two related concepts, attitudes and beliefs,
are frequently studied to understand human behavior. Beliefs, which
represent people's ideas about the way the world is or should be, are
commonly conceptualized as binary in nature, present or
absent\cite{FishbeinAjzen75}. Throughout their lives, people acquire new beliefs, and
sometimes, new beliefs replace old beliefs. In this way, people tend
to acquire a very large number of beliefs over the life course.  This
notion of {\em belief} in social psychology,  that is {\em binary} in nature, can be considered similar to the
notion of ``influence'' in computational social network analysis which is also {\em binary} in nature.

Attitudes, on the other hand, are ``latent predispositions to respond
or behave in particular ways toward attitude
objects''~\cite{AlwinScott96}. In contrast to beliefs, which are
largely cognitive in nature, attitudes, have a cognitive, affective,
and a behavioral component~\cite{Rokeach70}. Being subjective in
nature, attitudes can vary in strength
such that a person can hold a {\em very strong attitude or a weak attitude
toward} an object or concept, and thus attitude quantifies the strength of belief~\cite{Ajzen01, FishbeinAjzen75}. 
Individuals acquire attitudes through
experiences and exposure. In the case of exposure, a body of research
shows that {\em repeated exposure} to an object/idea increases the
likelihood that a person will adopt a more favorable attitude toward
it~\cite{Zajonc68}.  Thus {\em attitude} being non-binary can be thought of {\em
  strength of influence}. Motivated by these studies, we study the
problem of arriving at a mathematical model that captures the notion
of attitude resulting from information propagation in social networks.

Our first contribution is to define a mathematical model for measuring
attitude.  Within social networks, people are often subjected to {\em
  repeated exposures} to information such as an anti-vaccine message,
a pro-GMO message, or gun safety messaging.  It has been observed that
when an individual is exposed to a large number of, say, anti-vaccine
messages, this increases the probability that that person will adopt a
similar anti-vaccine attitude.  Based on this, we postulate that the
{\em strength of influence or attitude} of an individual, toward an
object/concept, can be captured by the number of times the
individual receives the information from its neighbors.  In other
words, if an already influenced individual is further provided with
the same/similar influencing information, then the latter reinforces
the learned belief of the individual, thus shaping and increasing
his/her {\em \att}. We use the number of reinforcements as a way to
quantify the \att.

Using this model, we define attitude of an individual and the total
attitude of the network as functions from $2^V$ to reals ($2^V$
denotes the power set of nodes $V$ of the network).  We denote the
function that captures the total attitude of the network with
$\sigma_{Att}(.)$ We study the computational complexity of the
function $\sigma_{Att}$ and provide efficient algorithms to
approximate it. We prove that this function is \#P-hard and 
 it is monotone and submodular.  We provide an $(\epsilon,
\delta)$-approximation algorithm for computing attitude with provable
guarantees.  We then formulate the {\em attitude maximization}
problem--find a seed set $S$ of size $k$ that will result in maximum
total attitude of the network. We first prove that the attitude
maximization problem is NP-hard.  Based on the monotonicity and
submodularity of \att, we propose a greedy algorithm that achieves a
$(1-1/e)$ approximation guarantee.

We further introduce the concept of \emph{actionable} \att. The
introduction of this concept is motivated by the fact that individuals
with higher attitude (strongly influenced) are likely to act according
to the attitude. This is particularly important in campaigns (such as
political or gun-safety messaging), where motivated and dedicated
volunteers are necessary to carry and spread the message (possibly
beyond the social network); and such volunteers are the ones who are
strongly influenced.  Our second major contribution is the study of the
underlying computational problem related to actionable attitude
maximization. We prove that though the function for computing
actionable attitude is not submodular, it is {\em approximately
  submodular}. Based on this we design efficient approximation
algorithms to maximize the actionable attitude in a network.

\section{Related Work} 
\label{sec:related}

Computational models of information diffusion in social networks is
introduced and formalized in the seminal works of Domingos and
Richardson~\cite{domingos:kdd01} and Kempe, Kleinberg and
Tardos~\cite{kempe:kdd03}. 
There are two widely-studied probabilistic
diffusion models: \emph{Independent Cascade} (IC) model and
\emph{Linear Threshold} (LT) model.  
Kempe {\it et
  al.}~\cite{kempe:kdd03} proved that the influence maximization
problem is NP-hard, and also proved that a greedy algorithm achieves a
$(1 - 1/e)$ approximation guarantee.  The approximation guarantee of
the greedy approach stems from the non-negativity, monotonicity and
submodularity of the influence function. 
  Since then several improvements have been
proposed to make the greedy algorithm more practical and scalable
~\cite{leskovec:kdd07,
  chen:kdd10,jung:icdm12,NTD16,Tang14,Tang15,
  Borgs14,GalhotraAroraRoy16}. 
Several variants of the influence maximization problem have been
studied in the literature, since the work of Kempe {\it et
  al.} such as topic-aware influence maximization and targeted influence maximization~\cite{LiZhangTan15,Chen:2015,li:privacy11,Barbieri:icdm12,
 Song:cikm16,Guo:cikm13,PSBP18,LuL12}.

Enhancements to the basic influence propagation model have been
proposed that take into account the opinions of
users~\cite{ZhangDinhThai13,GalhotraAroraRoy16,negative11}.
 Liu et al.~\cite{Liu:ICDM13,Liu:TKDD17} introduced PageRank based
diffusion model, as a generalization of the basic IC model.

These models do not capture the notion of
{\em attitude/strength of influence} that we seek to formalize.  
Aggarwal et al.~\cite{Aggarwal:SDM11} introduced a flow
  authority model to determine  assimilation of information in a
  network.  This model  differs from the Independent Cascade and does not capture the notion of attitude due to repeated activations. Consider a network where node 1 has a directed edge to node 2 and 3, and node 2 has a directed edge to node 3, and edge probabilities are 1.  Due to repeated activation, node 3 can receive information from nodes 1 and 2 and thus should have a higher attitude than nodes 1 and 2. However, in the flow-authority model all nodes will have equal probability of receiving ($p=1$) and does not distinguish node 3 from others whereas our proposed model will.  
 
In~\cite{Zhou:ICCS14}, the authors discussed the problem of maximizing
cumulative influence in a model where the same node can repeatedly
activate his/her neighbor within a given time interval. This is
realized by identifying a node to be newly activated in multiple
iterations of the diffusion process (even if the node, under
consideration may have been already activated). Such a model may lead
to divergence in the computation of objective function, and hence, the
computation is parameterized by a time interval.  This distinguishes our
model where only the newly activated nodes can alter the attitude of
his/her neighbor; which ensures the convergence of computation of our
objective function and allows the method to be step agnostic.


\section{Preliminaries}
\label{sec:background}


We describe the notation and definitions used frequently  in this paper.

\begin{definition}[Monotonicity \& Submodularity]
Let $V$ be a ground set and $f:2^V \rightarrow \mathbb{R}$ be a set function, where $2^V$ denotes the power set of $V$. We say that $f$ is {\em monotone} if $f(S) \leq f(T)$ when $S \subseteq T$. We say that $f$ is {\em submodular} if for every pair of sets $S$ and $T$ with $S \subseteq T$ and every $x \notin T$,
$f(S \cup \{x\}) - f(S) \geq f(T \cup \{x\} ) - f(T)$.
\end{definition}

We use $f(x|S)$ to denote the {\em marginal gain of $x$ with respect to $S$}, defined as $f(S \cup\{x\}) - f(S)$.

\begin{theorem}[Chernoff Bound]
Let $X_1,X_2...X_n$ be independent identically distributed random
variables taking value in the range $[0,1]$. $X=\sum_{i=1}^{n}X_i$. If
$\mu = \expectation[X]$, then for $\lambda \in (0,1)$, $P[|X-\mu| \geq
  \mu \lambda]\leq 2 exp(-\frac{\lambda^2}{2+\lambda}\cdot \mu)$
\end{theorem}

 A social network is modeled as a weighted
directed graph $G = (V, E)$ with parameters $p: e \in E \to [0, 1]$,
where V and E ($|V | = n$ and $|E| = m$) denote the set of nodes and
edges, respectively.  The function $p(e = (u, v))$ is the probability of  node $u$
influencing/activating node $v$.  This denotes probability that the
information is successfully transferred from $u$ to $v$. We first
recall the standard Independent Cascade model of information
diffusion.

\begin{definition}{[\textbf{IC-Model}]}
Information spreads via a random discrete process that begins at a set
$S$ called seed set. Initially at step zero, all nodes in $S$ are
activated/influenced.  In each step, each newly activated node $u$
attempts to activate/influence its {\em inactivated} neighbor $v$ with
probability $p(u, v)$.  The diffusion process terminates when no new
nodes are influenced in a step.
\label{sec:ic-model}
\end{definition}

Given a set of nodes $S$, let $\sigma(S)$ be the {\em expected number
  of nodes} that are influenced at the end of the diffusion process
when the seed set is $S$. 

\smallskip
\noindent {\sc Influence Maximization Problem}. Given a social network
$G = (V, E)$, and an integer $k>0$, find a seed set $S \subseteq V$ of
size $k$ such that $\sigma(S)$ is maximized.

\smallskip
Kempe {\it et al.}~\cite{kempe:kdd03} proved that the influence maximization
problem is NP-hard and showed that the function $\sigma(.)$ is
monotone and submodular. Based on this, they designed a
$(1-1/e)$-approximation algorithm for the influence maximization
problem.

\section{Modeling \attup}
\label{sec:attitude}

In this section, we provide a mathematical model and definition to
capture the notion of {\em attitude}.  

\begin{definition}{[\textbf{\attup-IC model (AIC)}]}
  The diffusion proceeds in discrete rounds starting from some set of
  seed-nodes $S$.  Initially, all non-seed nodes have the \att $0$ and every seed node starts with an attitude value of $1$.
   At each
  step, each newly influenced node $u$ tries to send information to
  each of its neighbor $v$ as per the edge probability $p(u,v)$. If
  $u$ succeeds, then $v$'s \att is incremented by $1$; and its status
  is changed to influenced if it is not already influenced.  When $u$
  succeeds in sending information $v$, we say that the edge $\langle
  u, v\rangle$ is activated. The process terminates when no new nodes
  are influenced in a step.
  \label{def:att-ic}
\end{definition}

Consider Figure~\ref{fig:exinfmax} and let seed
set is $S= \{a\}$. At step $t=0$, the
attitude of $a$ is 1, $a$ tries to send information to $b,c$,
succeeding with probability $1$. At $t=1$, the attitudes of $a,b,c$ are
$1$.  The newly activated nodes $b,c$
send information to their neighbors. Node $b$ succeeds and increments the
attitude of nodes $a,c$. Simultaneously, $c$ succeeds and increments
the attitude of nodes $a,b$. At $t=2$, the attitudes of $a,b,c$ are $3,
2, 2$ respectively. Since no new nodes are activated in this step,
the diffusion ends.

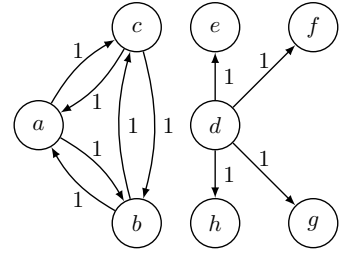
\begin{figure}[h]
\begin{center}
\resizebox{0.25\textwidth}{!}{
\begin {tikzpicture}[-latex,node distance =1.5cm and 1.5cm ,on grid ,
	semithick , scale=0.5,
	state/.style ={ circle ,top color = white , bottom color = white!10 ,
		draw,black , text= black , minimum width =0.75 cm}]
	\node[state] (B){$b$};
	\node[state] (A) [above left =of B] {$a$};
	\node[state] (C) [above right=of A] {$c$};
	\path (A) edge [bend left =15] node[above] {$1$} (C);
	\path (C) edge [bend left =15] node[below] {$1$} (A);
	\path (A) edge [bend left =15] node[above] {$1$} (B);
	\path (B) edge [bend left =15] node[below] {$1$} (A);
	\path (B) edge [bend left =15] node[right] {$1$} (C);
	\path (C) edge [bend left =15] node[right] {$1$} (B);
\end{tikzpicture}
\begin {tikzpicture}[-latex ,auto ,node distance =1.5cm and 1.5cm ,on grid ,
semithick ,
state/.style ={ circle ,top color = white , bottom color = white!10 ,
	draw,black , text= black , minimum width =.75 cm}]
\node[state] (D){$d$};
\node[state] (E) [above =of D] {$e$};
\node[state] (F) [above right =of D] {$f$};
\node[state] (G) [below right =of D] {$g$};
\node[state] (H) [below =of D] {$h$};
\path (D) edge [left =15] node[right] {$1$} (E);	
\path (D) edge [left =15] node[above] {$1$} (F);	
\path (D) edge [left =15] node[above] {$1$} (G);	
\path (D) edge [left =15] node[right] {$1$} (H);
\end{tikzpicture}
}
\end{center}
\caption{An example showing inf-max problem $\neq$ \att-max problem}
\label{fig:exinfmax}
\end{figure}

Note that, unlike in the standard information diffusion model, where each
activated node gets one chance to influence its \emph{un-influenced}
neighbors, in our model, each newly activated node tries to influence
all its neighbors irrespective of whether they are already influenced
or not.  Thus an activated node can receive information from a newly
activated node and this captures the notion of {\em repeated exposure}
or {\em reinforcement}, which, in turn, results in an increase of
 the recipient's attitude.


For any set $S \subseteq V$ of nodes, we use $\mathtt{Att}_v(S)$ to
denote the final attitude of node $v$ when the seed set is $S$.  Note
that this is a random variable and let $\expectation[\attof_v(S)],$
denote the expectation of $\mathtt{Att}_v(S)$. We define $\attin(S)$
as $\sum_{v \in V} \mathtt{Att}_v(S)$. The {\em total expected \att of the network}  resulting from
diffusion starting at seed $S$ is $\sigatt(S) =
\expectation[\attin(S)]$.  Observe that by l linearity of expectation,  $\sigatt(S) = \sum_{v \in V} \expectation[\attof_v(S)].$

By overloading notation, we often interpret $G$ as a distribution over
unweighted directed graphs, each edge $e=(u,v)$ is realized
independently with probability $p(u,v)$.  We write $g\sim G$ to denote
that an unweighted graph $g$ is drawn from this graph distribution
$G$.  Given a set of nodes $S\subseteq V$ and a graph $g$, we use
\begin{enumerate}
\item $R^S_g$ to denote the set of nodes reachable from
  $S$ in $g$.
  
\item $E^S_g = \{ e=(u,v) | u,v \in R^S_g$ and $e\in g \}$ is the set
  of \emph{activated edges} in $g$ due to diffusion from $S$. Let $E^S_{g, v}$ be the set of activated edges of the form $\langle ., v\rangle$.

\item $\attin_g(S)$ to denote the \att induced by $S$ in
  graph $g$ and is equal to $\sum_{v\in V} \attof_{g,v}(S)$,
  where $\attof_{g,v}(S)$ is the \att of $v$ in the graph $g$ computed as
  the number of activated incoming edges to $v$.
\end{enumerate}

We next prove a critical theorem that will be used in our subsequent
proofs.  Informally, this theorem states that the $\sigatt(S)$ is the
expected number of activated edges.

\medskip

\begin{theorem}\label{thm:actedge}
  If $g\sim G$ then for any $S\subseteq V$,
  $
    \sigatt(S)= |S| + \sum_{g\sim G} |E^S_g| \times Pr(g \sim G),
  $
  and
  $\expectation[\attof_v(S)] = \sum_{g\sim G} |E^S_{g,v}| \times Pr(g \sim G)$.
  \label{sec:sigatt-exp}
\end{theorem}
\begin{proof} 

Recall that, a node $u$ contributes to the \att of
its neighbor $v$, if $u$ is influenced and it is successful in
``passing'' on the influence to $v$ (irrespective of whether $v$ is
already influenced or not) via the directed edge $\langle u, v\rangle$. We refer to such an
edge as an activated edge.

Let $g$ be a graph drawn as per the distribution. Note that $g$ corresponds to a particular diffusion process.
In $g$, if a node $v$ is not reachable from $S$, it means $v$ is not
activated in that diffusion process, and its incoming edges, if any,
are not activated. Thus the \att of such a node is $0$. On the other
hand, if a node $v$ is reachable from $S$ in $g$, it means $v$ is
activated in the diffusion process. If $x$ is the number of incoming edges to $v$ in $G$, this means that
$v$ received information through its neighbors $x$ times. Thus the attitude of $v$ is $x$ in this diffusion process.
Thus  node $v$'s  \att is the number of
\emph{activated} incoming edges of $v$. Let $N(v)$ denote the number
of activated incoming edges of $v$. Then $\attin_g(S)$ is equal to
\[
\begin{array}{r@{\ =\ }l}
  \displaystyle\sum_{v\in V}\attof_{g,v}(S) &
  \displaystyle\sum_{v \in R^S_g }\attof_{g,v}(S) + \sum_{v \notin R^S_g }\attof_{g,v}(S)
  \\[0.5em]
  &
  \displaystyle |S| + \sum_{v \in R_g^S }N(v) + 0 = |S| +  |E^S_g|
  \end{array}
\]
The term $|S|$ is due to the fact that every seed node starts with an
attitude value of $1$.  This leads to
\[
\begin{array}{r@{\ = \ } l}
  \sigatt(S) & \expectation[\attin(S)]\\ & \expectation_{g\sim
    G}[\attin_g(S)] = |S|+ \displaystyle\sum_{g\sim G} |E^S_g| \times
  Pr(g \sim G)
\end{array}
\]

The second equality stated in the theorem follows from similar arguments. Let $g$ be a graph drawn as per the distribution. Observe that $\attof_{g,v}(S) = |E^S_{g,v}|$. This leads to:
\[
\begin{array}{r@{\ = \ } l}
 \expectation[\attof_v(S)] & \sum_{g\sim G} \attof_{g,v}(S) \times Pr(g \sim G)\\
 & \sum_{g\sim G} |E^S_{g,v}| \times Pr(g \sim G)
\end{array}
\]
\end{proof}

\subsection{Properties of \attup }
\label{sec:att-prop}

 In this section, we investigate several properties of the function $\sigatt(.)$. We first show that the $\sigatt$ is monotone and submodular

\begin{theorem}
Under the \attic model, $\sigatt(.)$ is a monotone, non-decreasing function
function.
\end{theorem}

\begin{proof} 
Let $g\sim G$ and $S\subseteq T\subseteq V$.  We observe $R^S_g
\subseteq R^T_g$ since $S\subseteq T$. Thus, $E^S_g \subseteq E^T_g$
and $|E^S_g| \leq |E^T_g|$. Therefore, $\sigatt(S) \leq \sigatt(T)$.
\end{proof}

\begin{theorem}
Under the \attic model, $\sigatt(.)$ is a submodular function.
\end{theorem}

\begin{proof}
Let $g\sim G$, $S\subseteq T\subseteq V$ and
$u\notin T$. Our objective is to prove that
\[
\begin{array}{l}
  \sigatt(S\cup\{u\}) - \sigatt(S) \\
   =   \displaystyle\sum_{g\sim G} (|E^{S \cup\{u\}}_g| - |E^S_g|) \times Pr(g\sim G) \\[.5em]
  \geq  \sigatt(T\cup \{u\}) - \sigatt(T) \\[.5em]
   =  \displaystyle\sum_{g\sim G} (|E^{T \cup\{u\}}_g| - |E^T_g|) \times Pr(g\sim G)
  \end{array}
\]
Since $Pr(g \sim G) \geq 0$, the proof obligation is
\[
\forall g\sim G\ \ |E^{S \cup \{u\}}_g| - |E^S_g| \geq |E^{T \cup \{u\}}_g| - |E^T_g|
\]

Observe that, 
\[
\begin{array}{l}
|E^{S \cup \{u\}}_g| - |E^{S}_g| = |E^{S \cup \{u\}}_g \setminus E^S_g|
\mbox{ and } \\
|E^{T \cup \{u\}}_g| - |E^{T}_g| = |E^{T \cup \{u\}}_g \setminus E^T_g|
\end{array}
\]
$R_g^S\subseteq R_g^T$ and 
$E^S_g\subseteq E^T_g$. 

For any $g\sim G$, if $e \in E^{T \cup \{u\}}_g \setminus E^T_g$ then
$e\notin E^T_g$ and $e\in E^{\{u\}}_g$. Since $E^S_g\subseteq E^T_g,
  e\notin E^S_g$. We know that $e\in E^{\{u\}}_g$ and thus $e\in E^{S \cup
    \{u\}}_g$. Therefore, $e \in E^{S \cup \{u\}}_g \setminus E^S_g$
  and thus $E^{T \cup \{u\}}_g \setminus E^T_g \subseteq E^{S \cup
    \{u\}}_g \setminus E^S_g$. This leads to $|E^{S \cup \{u\}}_g
  \setminus E^S_g| \geq |E^{T \cup \{u\}}_g \setminus E^T_g|$.
\end{proof}

%






The following result establishes the hardness of computing $\sigatt$.
\begin{theorem}\label{hard}
Under the \attic model, given $G=(V,E)$ and a seed $S\subseteq V$,
computing the values of the following is \#P-Hard: 1) $\sigatt(S)$, 2)
$\expectation[\attof_v(\cdot)], \forall v\in V$.
\end{theorem}

\begin{proof}
Let $\sigma(S)$ be the influence of $S$ under the IC
model. Computation of $\sigma(S)$ is known to be a \#P-Hard problem
~\cite{chen:kdd10}. Assume that there exists a function $A(G,S)$ that
computes $\sigatt(S)$. Let $a_1=A(G,S)$. Add a new vertex $v_{new}$ to
$G$. $\forall v\in V$, add an edge $(v, v_{new})$ and set
$p(v,v_{new})=1$. This results in graph $G'$. Let
$a_2=A(G',S)$. $a_2-a_1=\sum_{v\in V} P(S \texttt{ activates v}) =
\sigma(S)$. Therefore, $A$ can be used to compute
$\sigma(S)$. Similarly, let $A'(G,v)$ be a function that computes
$\expectation[\attof_v(S)]$. $A'(G',v_{new})$ will be able to compute
$\sigma(S)$ as $\expectation[\attof_{v_{new}}(S)]=\sigma(S)$. Similar arguments prove that computing $\expectation[\attof_v(\cdot)]$ is also \#P-hard.
\end{proof}

\subsection{\attup Computation}
\label{sec:att-rr}

From Theorem~\ref{hard}, it follows that computing $\sigatt(S)$
exactly is computationally infeasible.  In this section, we provide
efficient approximation algorithms to estimate $\sigatt(S)$.  Borgs
et. al.~\cite{Borgs14} introduced Reverse Influence Sampling (RIS),
which has been used to develop efficient Influence Maximization
algorithms~\cite{Tang14,Tang15,NTD16,Huang:2017}.  Using ideas from
these works, combining with Theorem~\ref{thm:actedge}, we introduce a
\emph{Reverse \attup Sampling} (RAS) technique.

Recall that $g$ denotes the un-weighted graph drawn from the random
graph distribution $G$. We write $g^T$ to denote the transpose of $g$.
The following lemma  and theorem establish the relationship between an edge
being activated by some nodes in any set $S \subseteq V$ and the
reachability of some node in $S$ from reverse of the same edges in
$g^T$; this relationship is key to the correctness of RAS technique.

\begin{lemma}
  Let $e=(x,y)$ be an arbitrary edge in $G$, $R^{\{x\}}_{g^{T}}$ be
  the set of nodes reachable from $x$ in $g^{T}$, where $g^T$ is the
  transpose of un-weighted graph $g$ drawn from random distribution
  $G$. Then for any $S\subseteq V$, $P[S \text{ activates } e \mbox{
      in } g] = P[S\cap R^{\{x\}}_{g^T}\neq \emptyset]$
\end{lemma}

Both events, $S \text{ activates } e \mbox{ in } g$ and $S\cap
R^{\{x\}}_{g^T}\neq \emptyset$ requires drawing $g$ from $G$ such that
there exists a path between some node in $S$ and node $x$ (from $S$
to $x$ in $g$ and $x$ to $S$ in $g^T$). The probability of occurrence
of such events are identical, as the probabilities of edges in $g$ and
their reverse in $g^T$ are equal.

The following theorem relates the $\sigatt(S)$ to reverse attitude sampling.
\begin{theorem}\label{thm:attitudeRAS}
Given a graph $G = (V, E)$, for any $S\subseteq V$, and for any $v \in
V$, let $\expectation(\attof_v(S))$ denotes the expected \att of v
induced by $S$. Then,
$\expectation(\attof_v(S)) = |InDegree(v)|\times
P_{g\sim G, e = (u,v)\sim E}[S \cap R^{\{u\}}_{g^T}~|~e\in g]
$ and
$\sigatt(S) = |S| + 
  |E|\times P_{g\sim G, e = (x,y)\sim E}[S \cap R^{\{x\}}_{g^T}~|~e\in
    g]$

\end{theorem}

\begin{proof}
With respect to  a set $S$ and a node $v$, we will define the random variable
\[
X^{(u,v)}_g =
\left
\{
\begin{array}{ll}
1 & \mbox{if } (u, v) \in E^S_g \\
0 & \mbox{otherwise}
\end{array}
\right.
\]
Therefore,  by Theorem~\ref{thm:actedge}, it follows that 
\[\expectation(\attof_v(S))  = \displaystyle\sum_{ (u, v) \in E}\expectation_{g\sim G}[X^{(u,v)}_g]. \]
Note that, 
\[
\begin{split}
\expectation_{g\sim G}[X^{(u,v)}_g]
&=
P_{g\sim G}[\exists w\in S.\ u \in R^{\{w\}}_g \land (u, v) \in g]\\
& =
P_{g\sim G}[\exists w\in S.\  w\in R^{\{u\}}_{g^T}  \land  (u, v)\in g]
\end{split}
\]

By linearity of expectation, we have:
\[
\begin{split}
\expectation(\attof_v(S)) &= \sum_{(u, v) \in E} \expectation_{g\sim G}[X^{(u,v)}_g] \\
&= \sum_{(u, v) \in E}   P_{g\sim G}[\exists w\in S.\  w\in R^{\{u\}}_{g^T}  \land  (u, v)\in g] \\
&= |\mbox{InDegree}(v)|\!\times\! P_{g\sim G, e = (u,v)\sim E}[S \cap R^{\{u\}}_{g^T}|e\in
g]
\end{split}
\]




We present the proof of the second equality. With respect to  a set $S$, we will define the random variable
$X^e_g =1$ if $e \in E^S_g$, otherwise it is zero.
Therefore, by Theorem~\ref{thm:actedge}, we have  $\sigatt(S)=\expectation_{g\sim G}[\attin_g(S)] = |S| + \displaystyle\sum_{e\in E}\expectation_{g\sim G}[X^e_g]$. Note that, 
\[
\begin{split}
    \expectation_{g\sim G}[X^e_g]
  &=
    P_{g\sim G}[\exists u\in S.\ x \in R^{\{u\}}_g \land e = (x, y) \in g]\\
  & =
    P_{g\sim G}[\exists u\in S.\  u\in R^{\{x\}}_{g^T}  \land  e=(x, y)\in g]
\end{split}
\]

By linearity of expectation, we have:\\
\[
\begin{split}
\sigatt(S) &= |S| +  \sum_{e\in E}\expectation_{g\sim G}[X^e_g] \\
&= |S| + \sum_{e\in E}     P_{g\sim G}[\exists u\in S.\  u\in R^{\{x\}}_{g^T}  \land  e=(x, y)\in g] \\
&= |S| + |E|\!\times\! P_{g\sim G, e\in E}[\exists u\in S.\  u\in R^{\{x\}}_{g^T}  \land  e=(x, y)\in g]
\end{split}
\]
\end{proof}

The above properties pave way for the RAS technique. We proceed by
introducing \emph{Random Reverse Reachable Set} in the context of the
\attic model.  Given a graph $G = (V, E)$, we construct Random Reverse
Reachable Set ($RR$) of nodes in $V$ as follows. Consider the
transpose of $G$, $G^T = (V, E^T)$, where the probability annotation
for any edge in $E$ remains unchanged in the reverse of that edge in
$E^T$.

We now describe a procedure to generate  \emph{Random Reverse Reachable Sets (RR Sets)}:\\
\noindent{{\bf Generate RR Set.} Randomly pick an edge $e = (v, u) \in E^T$.
Then with probability $p(e)$, add the node $u$ to $RR$. For any $u$ is
added to $RR$, for each outgoing edge from $u$ in $G^T$, add the
destination with corresponding edge probability. The process continues
till no node is added to $RR$. 

From Theorem~\ref{thm:attitudeRAS}, we obtain the following lemma.

\begin{lemma}
$\sigatt(S) = |S| + |E|\times P_{RR\sim \mathcal{R}}[S\cap RR \neq \emptyset]$
\label{lem:attitude}
\end{lemma}

\begin{algorithm}[t]
\scriptsize
{\small\sf}
\DontPrintSemicolon
\KwData{Graph $G=(V,E)$, $S\subseteq V$}
\Begin{
$\mathcal{R} = \mbox{Generate $\beta$ RR Sets using {\bf Generate RR Set}}$\;
$X= |\{RR\in \mathcal{R}~|~S\cap RR\neq \emptyset \}|$ \;

\KwRet {$\displaystyle\frac{|E| \cdot X}{\beta}$}\;
}
\caption{Estimate $\sigatt(S)$}
\label{algo:totalattitude}
\end{algorithm}

Lemma ~\ref{lem:attitude} allows us to design Algorithm
~\ref{algo:totalattitude} to estimate $\sigatt(S)$. In order to get a
good estimate, we will obtain a lower bound for $\beta$ in Algorithm
~\ref{algo:totalattitude}. Let $m=|E|$. Let $X_i$ be a random variable
that takes value $1$ if the $i$-th $RR$ Set contains an element of
$S$. Otherwise, $X_i=0$. Clearly each $X_i$ is independent and
$X=\sum_{i=1}^{\beta} X_i$. Note,
$\expectation[X]=\displaystyle\frac{\beta\sigatt(S)}{m}$

\[
\begin{array}{cl}
  & P[|\widehat{\sigatt}(S) - \sigatt(S)| \geq \epsilon \sigatt(S)] \\ 
= & P[|m \displaystyle\frac{X}{\beta} - \sigatt(S)| \geq \epsilon \sigatt(S)] \\
= & P[|X - \displaystyle\frac{\beta \sigatt(S)}{m}| \geq \epsilon \cdot \displaystyle\frac{\beta}{m} \sigatt(S)] \\
\leq & 2 exp(-\displaystyle\frac{\epsilon^2 \beta\sigatt(S)}{(2+\epsilon)m})
\end{array}
\]

The last inequality follows by applying Chernoff Bounds with $\lambda=\epsilon$. Let $\delta = 2 exp(-\displaystyle\frac{\epsilon^2 \beta\sigatt(S)}{(2+\epsilon)m})$. When $\beta \in \theta(\frac{m}{\epsilon^2\sigatt(S)}\cdot log(\frac{1}{\delta})$, Algorithm $~\ref{algo:totalattitude}$ estimates $\sigatt(S)$ within a relative error of $\epsilon$ with probability $1-\delta$.



%


\section{Attitude Maximization Problem}
\label{sec:attitude-maximization}

Having defined \attup under the \attic-model, a natural problem
arises: How do we find a set of users, who can influence the network
in a way that maximizes the \att of the network? We model this as
the \attup Maximization Problem:
\begin{problem}
\textsc{\attup Maximization Problem}: Given a graph $G=(V,E)$, a
number $k$, find $S\subseteq V$ of size at most $k$ such that
$\sigatt(S)$ is maximized.
\end{problem}

\begin{theorem}
  Under the \attic model, the \att maximization problem, i.e.,
  computing $\mathtt{argmax_{S\subseteq V,
      |S|\leq k}}\ \sigatt(S)$, is NP-hard.
  \label{thm:nphard}
\end{theorem}
\begin{proof} Our proof relies on reduction of influence
maximization problem (a known NP-Hard problem) to \att maximization
problem. 

We consider the influence maximization problem on directed Bi-partite
graphs (edges from left nodes to right nodes) with edge probabilities
1. That is, $G = (V, E)$, where $V = X \cup Y$, $X \cap Y = \emptyset,
E = \{(u, v) | u \in X, v \in Y\},$ and $\forall e \in E $ $p(e) =
1$. Kempe {\it et al.}~\cite{kempe:kdd03} proved that influence maximization
problem on such restricted class of graphs is also NP-hard.

We extend the bipartite graph $G$ to construct an instance $G' = (V',
E')$ for the attitude maximization problem, where $V' = V \cup Z, Z =
\{z_1, z_2, \ldots, z_{2|E|}\}$ and for each $y\in Y$, there
exists an edge to each $z\in Z$ with the edge probability $1$.

Suppose  that there is an algorithm for computing a set
$S\subseteq X$ of size $k$ that maximizes $\sigatt(S)$. If $L$ nodes
in the set $Y$ are influenced by $S$, then $\sigatt(S)\leq L\times
2|E| + |E|$.  (Each edge from an influenced node in $Y$ contributes to
the \att of each nodes in $Z$, and the overall \att of nodes in $Y$
can be at most $|E|$, the number of edges between $X$ and $Y$.)

Assume that $S$ does not induce maximum influence in $G$, i.e., there
exists some $S'\neq S$ for which $G$ is maximally influenced. In other
words, $S'$ influences at least $L+1$ nodes in $Y$. Therefore, if $S'$
is used as seed in $G'$, then it would have induced the overall \att
of nodes in $Z$ to be $(L+1)\times 2|E|$. This implies, $S'\neq S$ is
a set of size $|k|$ that maximizes $\sigatt(S')$ in $G'$, leading to a
contradiction.

Therefore, if any algorithm that  computes a set $S$ that maximizes \att in $G'$,
then $S$ must also maximize influence in $G$.
\end{proof}

Before we proceed to present an approximation algorithm for the
attitude maximization problem, we first prove that influence maximization
problem is different from the attitude maximization problem. In particular, we
prove that the optimal solution for the influence maximization problem
is not an optimal solution for the attitude maximization problem. Consider the from Figure~\ref{fig:exinfmax}. 
When $k=1$, the best seed set for the influence maximization is $\{d\}$ whereas the best seed set for the attitude maximization is any of $\{a\}, \{b\} $ or $\{c\}$. Thus,

\begin{theorem}
An optimal solution to the influence maximization problem is not an
optimal solution to the attitude maximization problem.
\end{theorem}



%
%

Nemhauser et. al. ~\cite{nemhauser} proved the greedy strategy to
maximize a non-decreasing, monotone, and submodular function outputs a
$(1-1/e)$-approximate solution. Recall that $\sigatt(\cdot)$ is in
fact a non-decreasing, monotone and submodular function. However, the
challenge lies in efficiently estimating $\sigatt(\cdot)$. Motivated
by this, we design a \textsc{RAS}-based approximation algorithm.

\begin{algorithm}[t]
\scriptsize
	{\small\sf}
	\DontPrintSemicolon
	\KwData{Graph $G=(V,E)$, $k$}
	\KwResult{Seed Set $S$ }
	\Begin{
		$\mathcal{R} = \mbox{Generate $\beta$ RR Sets using {\bf Generate RR Set}}$\;
		Mark all the sets in $\mathcal{R}$ as uncovered\;
		\While{$|S| \leq k$}{
		  Find $v$ that covers maximum uncovered sets in $\mathcal{R}$\;
		  Mark sets covered by $v$ as covered\;
		  Add $v$ to $S$\;
		}
		\KwRet {$S$}\;
	}
	\caption{$(1-1/e-\epsilon)$-approximate algorithm}
	\label{algo:greedyapprox}
\end{algorithm}

Algorithm~\ref{algo:greedyapprox} is our greedy algorithm for the attitude maximization problem. 
The algorithm works by generating $\beta$ random RR Sets. With the
goal now to find $S$ that covers the maximum RR Sets, the problem is
transformed to the Maximum Coverage problem. The greedy algorithm,
when applied to the Maximum Coverage problem, provides a
$(1-1/e)$-approximate solution. We have the following result on the approximation guarantee Algorithm~\ref{algo:greedyapprox}.
\begin{theorem}
When $\beta \in \theta (\frac{|E|(1+1\epsilon)}{\epsilon^2 \OPT} (log{n\choose k} - log(\delta) ))$, Algorithm~\ref{algo:greedyapprox} outputs a seed set $S_k$ such that
 \[\sigatt(S_k) \geq \left( 1 - \displaystyle\frac{1}{e} -\epsilon \right) \OPT\]
 with probability at least $1 -\delta$.
\end{theorem}
\begin{proof}
We will prove that the algorithm produces a $(1-1/e-\epsilon)$-approximate solution with high probability.

First, we derive the bound for $\beta$ that is sufficient for
estimating $\sigatt(S)$ within a pre-specified error margin
$\epsilon$, in the context of computing the maximal overall \att.

Consider any $S\subseteq V$ of size $k$. Let $X$ be the cardinality of
$\{RR \in \mathcal{R}| RR\cap S\neq \phi\}$. $\widehat{\sigatt}(S) =
|E|\times \displaystyle\frac{X}{\beta}$ is a an estimate for $\sigatt(S)$.  Let
$\mu = \displaystyle\frac{\beta\cdot\sigatt(S)}{|E|}$ and $\OPT$ be the maximum
expected \att induced by any set of size $k$. 

\[
\begin{array}{cl}
 & P\left[|\widehat{\sigatt}(S) - \sigatt(S)|\geq \displaystyle\frac{\epsilon \OPT}{2}\right] \\
= & P\left[|E|\cdot\displaystyle\frac{X}{\beta} - \sigatt(S)|\geq \displaystyle\frac{\epsilon \OPT}{2}\right] \\
= & P\left[| \displaystyle\frac{X}{\beta} - \displaystyle\frac{\sigatt(S)}{|E|}|\geq \displaystyle\frac{\epsilon \OPT}{2|E|}\right] \\
= & P\left[| X - \mu|\geq \displaystyle\frac{\epsilon \OPT\cdot \beta}{2|E|}\right] \\
= & P\left[| X - \mu|\geq \displaystyle\frac{\epsilon \OPT\cdot \beta\sigatt(S)}{2\sigatt(S)|E|}\right] \\
\end{array}
\]

We apply Chernoff Bounds with $\lambda=\displaystyle\frac{\epsilon
 \OPT}{2\sigatt(S)}$,
\[
\begin{array}{cl}
 & P\left[|X - \mu|\geq \lambda \mu \right]
 < 2 exp\left(-\displaystyle\frac{\lambda^2}{2+\lambda} \mu \right) \\
  
 =\ & 2 exp\left(-\displaystyle\frac{\epsilon^2 (\OPT)^2}{(2+\lambda) \times 4(\sigatt(S))^2} \mu \right) \\
  =\ & 2 exp\left(-\displaystyle\frac{\epsilon^2 (\OPT)^2}{(2+\lambda) \times 4(\sigatt(S))^2} \displaystyle\frac{\beta\cdot\sigatt(S)}{|E|} \right) \\
  =\ & 2 exp\left(-\displaystyle\frac{\epsilon^2 (\OPT)^2}{(2+\lambda) \times 4\sigatt(S)} \displaystyle\frac{\beta}{|E|} \right) \\
 =\ & 2 exp\left(-\displaystyle\frac{\epsilon^2 (\OPT)^2}{|E|(8\sigatt(S)+2\epsilon\OPT)} \beta \right) \\
 \leq \ & 2 exp\left(-\displaystyle\frac{\epsilon^2 (\OPT)^2}{|E|(8\OPT+2\epsilon\OPT)} \beta \right) \\
 =\ & 2 exp\left(-\displaystyle\frac{\epsilon^2 \OPT}{|E|(8+2\epsilon)} \beta \right)
\end{array}
\]

The inequality follows from $\OPT \geq \sigatt(S)$. We would like the probability of this event to be at most
$\displaystyle\frac{\delta}{{n\choose k}}$. Proceeding further,

\[
\begin{array}{rcl}
2 exp\left(-\displaystyle\frac{\epsilon^2 \OPT}{|E|(8+2\epsilon)} \beta \right) &\leq& \displaystyle\frac{\delta}{{n\choose k}} \\
-\displaystyle\frac{\epsilon^2 \OPT}{|E|(8+2\epsilon)} \beta &\leq& log\left( \displaystyle\frac{\delta}{2{n\choose k}} \right) \\
\end{array}
\]

This implies that
\[\beta \geq \frac{|E|(8+2\epsilon)}{\epsilon^2 \OPT} \left[ log(2) + log{n\choose k} - log(\delta) \right]\]

\[
\begin{array}{rcl}
\beta & \geq & -\displaystyle\frac{|E|(8+2\epsilon)}{\epsilon^2 \OPT} log\left( \displaystyle\frac{\delta}{2{n\choose k}} \right) \\[1em]
& = & -\displaystyle\frac{|E|(8+2\epsilon)}{\epsilon^2 \OPT} \left[ log(\delta) - log(2) - log{n\choose k} \right] \\[1em]
& = & \displaystyle\frac{|E|(8+2\epsilon)}{\epsilon^2 \OPT} \left[ log(2) + log{n\choose k} - log(\delta) \right]
\end{array}
\]

Now that we have a lower bound for $\beta$, we can use the union bound
to show that this number of $RR$ sets is sufficient to ensure that
\emph{all} sets of size $k$ is within $\epsilon\cdot \OPT/2$ with
probability at least $1-\delta$. More precisely,
\begin{equation*}
P\left[\forall S, |S|=k, |\widehat{\sigatt}(S) - \sigatt(S)|\geq \displaystyle\frac{\epsilon \OPT}{2}\right] \leq \delta
\end{equation*}

Finally we relate the output of~\ref{algo:greedyapprox} with the optimal solution.
Let $S_{k}$ be the output of Algorithm ~\ref{algo:greedyapprox} and $S'$ the optimal solution to the coverage problem. Let $\Delta^*, \Delta', \Delta^k$ be the number of $RR$ sets covered by the $S^*, S', S_k$ respectively. With probability at least $1-\delta$,

\begin{IEEEeqnarray*}{lCl}
|\sigatt(S_k) - \widehat{\sigatt}(S_k)| &\leq & \displaystyle\frac{\epsilon \OPT}{2} \\
\sigatt(S_k) - \widehat{\sigatt}(S_k) &\geq & \displaystyle\frac{-\epsilon \OPT}{2} \\
\end{IEEEeqnarray*}
\begin{IEEEeqnarray*}{lCl}
\sigatt(S_k) &\geq & \widehat{\sigatt}(S_k)  - \displaystyle\frac{\epsilon \OPT}{2} \\
&  \geq & \displaystyle\frac{|E|}{\beta}\left( 1 - \displaystyle\frac{1}{e} \right) \Delta'  - \displaystyle\frac{\epsilon \OPT}{2} \\
& \geq & \displaystyle\frac{|E|}{\beta}\left( 1 - \displaystyle\frac{1}{e} \right) \Delta^*  - \displaystyle\frac{\epsilon \OPT}{2} \\
& \geq & \left( 1 - \displaystyle\frac{1}{e} \right) \widehat{\sigatt}(S^*)  - \displaystyle\frac{\epsilon \OPT}{2} \\
& \geq & \left( 1 - \frac{1}{e} \right) \left(1 - \displaystyle\frac{\epsilon}{2} \right)\OPT  - \displaystyle\frac{\epsilon \OPT}{2} \\
& = & \left( 1 - \displaystyle\frac{\epsilon}{2} - \displaystyle\frac{1}{e} + \displaystyle\frac{\epsilon}{2e} - \displaystyle\frac{\epsilon}{2} \right) \OPT\\
& \geq & \left( 1 - \displaystyle\frac{1}{e} -\epsilon \right) \OPT \\
\end{IEEEeqnarray*}

Thus, Algorithm ~\ref{algo:greedyapprox} outputs $\left( 1 - \displaystyle\frac{1}{e} -\epsilon \right)$-approximate solution with probability at least $1-\delta$.

\end{proof}

         



\section{ATTITUDE to \EffUpperCase}
\label{sec:attMinusInfluence}

As noted in
the introduction,  nodes with high influence are likely to act based on
their influence, and in some scenarios it is desirable to be able to spread information that results in such highly influenced
individuals. Motivated by this, we introduce a notion called {\em actionable attitude} that attempts to increase the total attitude of nodes with ``high enough attitude'', as opposed to the total attitude of all the nodes. For this, we need to understand and formulate the concept of high  enough attitude. Consider a network in which many nodes have an attitude value close to $2.5$ and a few nodes having an attitude more than $5$ (with respect to a certain seed set).  For this network,  a value of $5$ can be considered high, whereas for a network with most nodes having an attitude value of more than $7$, a value of $5$ is low. This suggests that the notion of high enough attitude is {\em relative} and depends on the structure of the network and the underlying influencing mechanisms. Thus, a way to formulate this notion is to incorporate the influence propagation. Consider a concrete instantiation of a diffusion process.
There are certain nodes that are barely influenced, they  receive the information once and thus their attitude is 1. However,  there exist certain nodes whose opinions have been reinforced due to multiple exposures. Comparatively these nodes can be thought of having higher attitude than the nodes that receive information only once. We refer to the attitude of these individuals in the network as \emph{actionable attitude}. Thus if the goal is to maximize this \emph{actionable attitude}, then we should discard the collective attitude of nodes that are barely influenced.  This leads us to the following definition.

\begin{definition}{[\textbf{\Eff}]}
We define \Eff induced by a given seed set $S$ as $\ExpEff(S) =
\sigatt(S) - \sigma(S)$.
\label{def:Effective Attitude}
\end{definition}

\begin{problem}
\label{prob:attMinusInf}
\textsc{\Eff Maximization Problem}: Given a graph $G=(V,E)$ and 
$k$, find $S\subseteq V$ of size at most $k$ such that $\ExpEff(S) $ is
maximized.
\end{problem}
We first show that the function $\ExpEff(\cdot)$ is a monotone function but not submodular.
\begin{theorem}
Under the \attic model, $\ExpEff(.)$ is a monotone, non-decreasing function
function
\end{theorem}

\begin{proof}
	Let $g\sim G$ and $S\subseteq T\subseteq V$.  We
	observe $|S| \leq |T|$ and $R_g(S) \subseteq R_g(T)$ since $S\subseteq T$. Thus, $E^S_g \subseteq E^T_g$ and $|E^S_g| \leq |E^T_g|$. For the subgraph $g' = (V' , E')$ induced by $R_g^T \backslash R_g^S$, $|E'| \geq |V'| - 1$ Therefore, $\ExpEff(S) = (|S| + |E^S_g| - R_g^T) \leq (|T| + |S| + |E^S_g| - R_g^T + |E'| - |V'|) = \ExpEff(T) $.  
	
Let $g\sim G$ and $S\subset T\subseteq V$.  We observe $|S| < |T|$ and
$R_g(S) \subseteq R_g(T)$ since $S\subset T$. Thus, $E^S_g \subseteq
E^T_g$ and $|E^S_g| \leq |E^T_g|$. For the subgraph $g' = (V' , E')$
induced by $R_g^T \backslash R_g^S$, $|E'| \geq |V'| - 1$ Therefore,
$\ExpEff(S) = (|S| + |E^S_g| - |R_g^S|) \leq (|S| + |E^S_g| - |R_g^S|
+ |E'| - |V'| + 1) \leq (|T| + |E^T_g| - |R_g^T|)= \ExpEff(T) $.
\end{proof}

\begin{theorem}
	Under the \attic model, $\ExpEff(.)$ is not submodular.
\end{theorem}

\begin{proof} 
	Consider the following graph $G$ with each edge probability $1$.  Note
	that, there exists exactly one $g \sim G$, which is the graph itself.

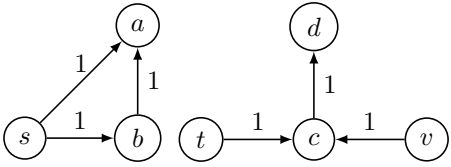
\begin{figure}[h]
\begin{center}
	\begin{tikzpicture}[-latex ,auto ,node distance =1.5 cm and 1.5 cm ,on grid ,
	semithick ,
	state/.style ={ circle ,top color = white , bottom color = white!10 ,
		draw,black , text= black , minimum width =0.5 cm}]
	\node[state] (S){$s$};
	\node[state] (B) [right =of S] {$b$};
	\node[state] (A) [above =of B] {$a$};
	\path (S) edge [left =15] node[above] {$1$} (B);
	\path (S) edge [left =15] node[above] {$1$} (A);
	\path (B) edge [left =15] node[right] {$1$} (A);

\end{tikzpicture}
\begin{tikzpicture}[-latex ,auto ,node distance =1.5cm and 1.5cm ,on grid ,
semithick ,
state/.style ={ circle ,top color = white , bottom color = white!10 ,
	draw,black , text= black , minimum width =0.5 cm}]
\node[state] (C){$c$};
\node[state] (D) [above =of C] {$d$};
\node[state] (T) [left =of C] {$t$};
\node[state] (V) [right =of C] {$v$};
\path (T) edge [left =15] node[above] {$1$} (C);	
\path (V) edge [left =15] node[above] {$1$} (C);	
\path (C) edge [left =15] node[right] {$1$} (D);	
\end{tikzpicture}
\end{center}
\caption{An example demonstrating $\ExpEff(.)$ is not submodular}
\label{fig:effnonsubmodular}
\end{figure}
	
Let $S = \{s\}, T = \{s, t\}$. $S \subseteq T$ and $v \notin
T$. Observe that, $\ExpEff(S) = ( |\{s\}| + |\{(s, a), (s, b), (b,
a)\}|) - |\{s ,a ,b\}|= 4 - 3 = 1$ and $\ExpEff(T) = (|\{s, t\}| +
|\{(s, a), (s, b), \linebreak (b, a), (t, c), (c, d)\}|) - |\{s, a, b,
t, c, d\}| = 7 - 6 = 1$. Similarly, $\ExpEff(S \cup \{v\}) = ( |\{s,
v\}| + |\{(s, a), (s, b), (b, a), (v, c), (c, d)\}|) - |\{s, v, a ,b,
c, d\}|= 7 - 6 = 1$ and $\ExpEff(T\cup \{v\}) = (|\{s, t, v\}| +
|\{(s, a), (s, b),(b, a), (t, c), (c, d), (v, c)\}|) - |\{s, a, b, t,
c, d, v\}|= 9 - 7 = 2$. Therefore, $\ExpEff(v| S) = \ExpEff(S \cup
\{v\}) - \ExpEff(S) = 1- 1 = 0 $ and $\ExpEff(v| T) = \ExpEff(T \cup
\{v\}) - \ExpEff(T) = 2- 1 = 1 $. Since $\ExpEff(v| S) < \ExpEff(v|
T)$, $\ExpEff(.)$ is not submodular.
\end{proof}


Note that $\sigatt(\cdot)$ and $\sigma(\cdot)$ are very closely related as they rely on the same diffusion process. Using this we show that the actionable attitude function $\ExpEff(.)$ is {\em approximately submodular} ~\cite{JMLR:v9:krause08a}. 
\begin{definition}
A set function $f$ is {\em $\Delta$\mbox{-}approximate submodular} if for
every pair of sets $S$ and $T$ with $S \subseteq T$ and every $x
\notin T$, $f(x|S) \geq f(x|T) -
\Delta$.
\end{definition}

Note that for submodular functions $\Delta$ is zero.  We show that the
unction $\ExpEff(\cdot)$ is $\Delta$-approximate submodular, where
\emph{$\Delta$ is the expected maximum degree of the graph}, where
each edge $\langle u, v\rangle$ is kept with probability $p(u, v)$.

\begin{theorem}
Given a graph $G=(V,E)$ let $deg_G(v)$ denote the outdegree of any $v\in
V$. Then, $\forall S\subset T\subseteq V$ and $\forall v\notin T$,
$\sigattinf(v|S)\geq \sigattinf(v|T)- \expectedDeg$.
\end{theorem}
\begin{proof}
Let $f(v|S) = [(|E^{S \cup\{v\}}_g| + |S| + 1) - |R^{S \cup\{v\}}_g|]
- [(|E^{S}_g| + |S| - |R^{S}_g|]$.  Our objective is to prove that
$\sigattinf(v|T) - \sigattinf(v|S) = \displaystyle\sum_{g\sim G}
f(v|T) \times Pr(g\sim G) - \displaystyle\sum_{g\sim G} f(v|S) \times
Pr(g\sim G) \leq \displaystyle\sum_{g\sim G} \deg_g(v) \times Pr(g\sim
G)$
	
Since $Pr(g \sim G) \geq 0$, the proof obligation is 
\[ 
\forall g\sim G\ f(v|T) - f(v|S) \leq \deg_g(v)
\]
 We consider 3 cases. 

\noindent
{\em Case 1. }  $R_g^{v} \cap R_g^T =
 \emptyset$. In this case
 \[ f(v|S) = (|E_g^v| + 1) - |R_g^v| = f(v|T).\]
 Thus,
 $f(v|T) - f(v|S) = 0 \leq \deg_g(v).$\\
 
 \noindent{\em Case 2.}  $R_g^{v} \cap R_g^T
 \neq \emptyset, R_g^{v} \cap R_g^S = \emptyset$. In this case 
 \[f(v|S) = (|E_g^v|
 + 1) - |R_g^v|\]
 and 
 \begin{eqnarray*}
  f(v|T) & = & \{[|E_g^T| + |E_g^v| - |E_g^T \cap E_g^v|\\
  &  &  + (|T| + 1)] - [|R_g^T| + |R_g^v| - |R_g^T \cap R_g^v|]\} \\
 &  & -  [|E_g^T| + |T|- |R_g^T|] = (|E_g^v| + 1 - |R_g^v|)\\
 & &  + (|R_g^T \cap
 R_g^v| - |E_g^T \cap E_g^v|).
 \end{eqnarray*}
 
 For the subgraph $g' = (V', E')$
 induced by $R_g^T \cap R_g^v \backslash (T \cup \{v\}) $, $|E'| \geq
 |V'| - 1$. Thus $(|R_g^T \cap R_g^v| - |E_g^T \cap E_g^v|)$ reaches
 its maximum value $\deg_g(v)$ when $E_g^T \cap E_g^v =
 \emptyset$. Thus, $f(v|T) - f(v|S) \leq \deg_g(v).$\\ 
 
 \noindent{\em Case 3.} 
 $R_g^{v} \cap R_g^S \neq \emptyset$. In this case, 
 \begin{eqnarray*}
 f(v|S) & = &  (|E_g^v| + 1-
 |R_g^v|) + (|R_g^S \cap R_g^v| - |E_g^S \cap E_g^v|)\\
  f(v|T) &  = & 
 (|E_g^v| + 1- |R_g^v|) + (|R_g^T \cap R_g^v| - |E_g^T \cap
 E_g^v|)
\end{eqnarray*}
Therefore, 
\[
 f(v|T) - f(v|S) = |(R_g^T \backslash R_g^S) \cap R_g^v| -
 |(E_g^T \backslash E_g^S) \cap E_g^v|.
 \]
  For the subgraph $g' = (V',
 E')$ induced by $(R_g^T \backslash R_g^S) \cap R_g^v$, $|E'| \geq
 |V'| - 1$. Thus $|(R_g^T \backslash R_g^S) \cap R_g^v| - |(E_g^T
 \backslash E_g^S) \cap E_g^v|$ reaches its maximum value $\deg_g(v)$
 when $(E_g^T \backslash E_g^S) \cap E_g^v = \emptyset$.\\ Thus,
 $f(v|T) - f(v|S) \leq \deg_g(v).$\\
\end{proof}

This leads to following theorem.
\begin{theorem}\label{thm:almostSub}
The function $\ExpEff(\cdot)$ is $\Delta$-approximate submodular, where
$\Delta$ is the expected max degree of the graph.
\end{theorem}

Using this we first show that a greedy algorithm for actionable
attitude maximization problem gives a $(1-1/e)$ approximation
algorithm with an additive error of $\Delta$. The greedy algorithm
starts with an empty set $S_0$. During the iteration $i$, it picks an
element $v$ such that $\ExpEff(S_{i-1} \cup\{v\}) - \ExpEff(S_{i-1})$
is maximized.  Let $S^*$ is the optimal solution to the actionable
attitude maximization problem and let $S_k$ be the seed set produced
by the greedy algorithm
\begin{theorem}
$\ExpEff(S_k) \geq (1-1/e)\ExpEff(S^*) - (k-1)\Delta$.
\end{theorem}

\begin{proof}
Let $S^*=\{e_1,e_2..,e_k\}$ be the optimum solution. 

\begin{align*}
   \sigattinf(S^*)\leq \sigattinf(S_i\cup S^*) = \sigattinf(S_i) + \sigattinf(S^*|S_i) \\
   = \sigattinf(S_i) + \sigattinf(e_1|S_i) + \sigattinf(e_2|S_i\cup \{e_1\}) + \\
    \sigattinf(\{e_3, e_4..e_k\}|S_i\cup \{e_1,e_2\}) \\
   \leq  \sigattinf(S_i) + \sigattinf(e_1|S_i) + \sigattinf(e_2|S_i) + \Delta + \\
   \sigattinf(\{e_3, e_4..e_k\}|S_i\cup \{e_1,e_2\}) \\
   \leq  \sigattinf(S_i) + \sum_{e\in S^*\setminus S_i}\sigattinf(e|S_i) + (k-1)\Delta \\
   \leq  \sigattinf(S_i) + k\sigattinf(S_{i+1}) - k\sigattinf(S_i) + (k-1)\Delta
\end{align*}

By subtracting $\sigattinf(S^*)$ on both sides,  rearranging terms, and solving the resulting recurrence we obtain

\begin{eqnarray*}
        \sigattinf(S_{i+1}) - \sigattinf(S^*) & \geq &\\
      (1-\frac{1}{k})(\sigattinf(S_{i}) - \sigattinf(S^*))  - (1-\frac{1}{k})\Delta\\
\end{eqnarray*}
Solving this recurrence, we get:

\begin{eqnarray*}
        \sigattinf(S_{k}) - \sigattinf(S^*)
    &\geq &(1-\frac{1}{k})^k (- \sigattinf(S^*)) \\
     & & 
- (k-1)\Delta \\
\end{eqnarray*}
hat $\sigattinf(S_{k})  \geq \left(1-\frac{1}{e}\right) \sigattinf(S^*)  - (k-1)\Delta$.
\end{proof}

\begin{algorithm}[t]
\scriptsize
	{\small\sf}
	\DontPrintSemicolon
	\KwData{Graph $G=(V,E)$, $S \subseteq V$, $k$}
	\Begin{
		\ForEach{$v \in V$}
		{$\mathcal{R} _v$ = Generate $a \times Indegree(v)$ RR graphs from v\\
	      	\ForEach{$g^T \in \mathcal{R} _v$}
		          {
		      	      $c_{g^T}^v(S) =$ the number of edges from v that reaches $S$ in $g^T$ - 1}}
	 
		\KwRet {$ \displaystyle \sum_{v \in V} \frac{\displaystyle \sum_{g^T \in \mathcal{R} _v} c_{g^T}^v(S)}{|\mathcal{R} _v|}  $}\;
	}
	\caption{Estimate $\ExpEff$}
	\label{algo:estimateEffectivAttitude}
\end{algorithm}

The greedy algorithm runs in polynomial time; however it is not
scalable.  As has been done for influence maximization~\cite{Borgs14}
and attitude maximization (Section~\ref{sec:attitude-maximization}),
we design a more efficient algorithm based on RR sets. However,
 the RR set based algorithms for those maximization problems do
not easily translate to the case of actionable attitude maximization.
The RR set based algorithm for influence maximization randomly picks a
vertex $v$ and generates a RR graph from $v$ whereas RR set based
algorithm for attitude maximization starts with picking an edge $e$
uniformly at random. For influence maximization problem it is critical
that each vertex is picked uniformly at random and for attitude
maximization, it is critical that each edge is picked uniformly at
random.  Note that randomly picking a vertex does not imply a random
choice of edge and vice versa. Since the function $\ExpEff(\cdot)$ is
the difference between attitude and influence, neither of these RR set
based methods can be translated for actionable attitude
maximization. 
We need a mechanism to generate RR sets using which we
can estimate both $\sigma$ and $\sigatt$.  Instead of randomly picking
a vertex or edge in the network, we generate a sufficient number of RR
graphs for each vertex $v$. 

 Let $F_g^S(v)$ be the number of edges from $v$ that reaches $S \in g^T$,
 $\mathcal{R} _v$ be the set of RR graphs from $v$, and $T_g^S(v)$ be
 the number of edges to $v$ that are reachable from $S \in g$.
 
\begin{theorem}\label{thm:actatt1}
	  Given a graph $G = (V, E)$, for any $S\subseteq V$. $\ExpEff(S) = \displaystyle \sum_{v \in V} \displaystyle \sum_{g^T \in \mathcal{R}_v} P(g) \times max \{F_g^S(v) - 1, 0\}$	
	\label{thm:estimateEffectivAttitude}
\end{theorem}

\begin{proof}
With respect to  a set S, we will define the random variable
\[
Inf_v(S) =
	\left
	\{
	\begin{array}{ll}
	1 & \mbox{if } v \in R^S_g \\
	0 & \mbox{otherwise}
	\end{array}
	\right.
\]
Then,
\[
\begin{array}{r@{\ = }l}
\ExpEff(S) & \expectation \left[ \displaystyle \sum_{v \in V} Att_v(S)\right] 
		-  \expectation \left[ \displaystyle \sum_{v \in V} Inf_v(S)\right] \\[0.5em]
        & \displaystyle \sum_{v \in V} \expectation \left[ Att_v(S) - Inf_v(S) \right]\\[0.5em]
        & \displaystyle \sum_{v \in V}  \displaystyle \sum_{g \sim G} P(g) \times   \left[ Att_v(S) - Inf_v(S) \right]\\[0.5em]
        & \displaystyle \sum_{v \in V}  \displaystyle \sum_{g \sim G} P(g) \times max \{T_g^S(v) - 1, 0 \} \\[0.5em]
        & \displaystyle \sum_{v \in V}  \displaystyle \sum_{g \sim G} P(g) \times max \{F_{g^T}^S(v) - 1, 0 \}\\[0.5em]
         & \displaystyle \sum_{v \in V}  \displaystyle \sum_{g^T \in \mathcal{R}_v} P(g^T) \times max \{F_{g^T}^S(v) - 1, 0 \}
\end{array}
\]
\end{proof}

\begin{algorithm}[t]
\scriptsize
	{\small\sf}
\DontPrintSemicolon
\KwData{Graph $G=(V,E)$, $k$}
\KwResult{Seed Set $S$ }
	\Begin{
		\ForEach{$v \in V$}
		{$\mathcal{R} _v$ = Generate $a \times Indegree(v)$ RR graphs from v\\
			\ForEach{$g^T \in \mathcal{R} _v$}
			{\ForEach{$u \in g^T$}
				{$c_{g^T}^v(u) =$ the number of edges from v that reaches $u$ in $g^T$ - 1}}}
	    \ForEach{$u \in V$}
	    { $c(u) =  \displaystyle \sum_{v \in V} \frac{\displaystyle \sum_{g^T \in \mathcal{R} _v} c_{g^T}^v(u)}{|\mathcal{R} _v|}  $	}
		\While{$|S| \leq k$}{
			$v^* = arg \max\limits_{ u \in V \char`\\ S} c(u) $
			\;
		    $S = S \cup \{v^*\}$\;
		    \ForEach{$v \in V$}
		    	{\ForEach{$g^T \in \mathcal{R} _v$}
		    	 { Remove $v^*$ and all associated edges from $g^T$ \;
		    	 	\ForEach{$u \in g^T$}
		    		{compute $c_{g^T}^v(u)$\; }}}
		}
		\KwRet {$S$}\;
	}
	\caption{Find Best Seed Set for $\ExpEff (\cdot)$}
	\label{algo:bestSEffectivAttitude}
\end{algorithm}

\begin{theorem}\label{thm:actatt2}
Given a graph $G = (V, E)$, for any $S\subseteq V, u \in V$, the following
holds:  $\ExpEff(u | S)$ is equal to
\[
\begin{array}{l}
\displaystyle \sum_{v \in V}\!\displaystyle \sum_{g^T \in
  \mathcal{R}_v}\!\! P(g) \cdot \left[\mathtt{max} \{F_g^{S \cup \{u\}}(v) - 1,
  0\} \right.\\
\ \ \ \ \ \ \ \ \ \ \ \ \ \ \ \ \ \ \ \ \ \ \ \ \ \ \ \  \ \ \ \ \ \ \ 
\ \ \ \ \ \ \ 
\left.- \mathtt{max} \{F_g^S(v) - 1, 0\} \right]
\end{array}
\]
\label{thm:bestSeedUpdateEffectivAttitude}
\end{theorem}

\begin{proof}
\[
\begin{array}{l}
\ExpEff(u | S) = \left[ \sigatt ( S \cup \{u\}) - \sigma (S \cup \{u\})\right] - \left[ \sigatt (S) - \sigma (S)\right]
\\[1em]
= \displaystyle \sum_{v \in V} \expectation \left[ Att_v(S \cup \{u\})\right.\\[-1em]
\left.\ \ \ \ \ \ \ \ \ \ \ \ \ \ \  - Inf_v(S \cup \{u\}) - ( Att_v(S) - Inf_v(S) )\right] 
\\[1em]
= \displaystyle \sum_{v \in V}  \displaystyle \sum_{g \sim G} P(g) \\
\ \ \ \ \ \ \ \ \ \ \ \ \ \times \left[ max \{T_g^{S \cup \{u\}}(v) - 1, 0 \} -  max \{T_g^S(v) - 1, 0 \} \right]
\\[1em]
= \displaystyle \sum_{v \in V}  \displaystyle \sum_{g \sim G} P(g) \\
\ \ \ \ \ \ \ \ \ \ \ \ \ \times \left[ max \{F_g^{S \cup \{u\}}(v) - 1, 0 \} -  max \{F_g^S(v) - 1, 0 \} \right]
\\[1em]
= \displaystyle \sum_{v \in V}  \displaystyle \sum_{g^T \in \mathcal{R}_v} P(g^T) \\
\ \ \ \ \ \ \ \ \ \ \ \ \ \times \left[ max \{F_g^{S \cup \{u\}}(v) - 1, 0 \} -  max \{F_g^S(v) - 1, 0 \} \right]
\end{array}
\]
\end{proof}

Using the above two theorems, we can prove that
Algorithm~\ref{algo:bestSEffectivAttitude} is an approximation
algorithm for the actionable attitude maximization problem. Let $S^*$
be an optimal solution and let $S_k$ be the set produced by
Algorithm~\ref{algo:bestSEffectivAttitude}.

\begin{theorem}
In algorithm~\ref{algo:bestSEffectivAttitude} if $a$ is $O(1/\epsilon^2 \log n/\delta)$, then 
\[\Pr[\ExpEff(S_k) \geq (1-1/e-\epsilon)\ExpEff(S^*) - (k-1)\Delta] \geq \delta\]
\end{theorem}

We can prove the above theorem using Theorems~\ref{thm:actatt1}
and~\ref{thm:actatt2} and techniques used to establish the guarantee
on RR set based algorithm for the attitude maximization problem. We
omit the details. Note that in this algorithm, as opposed to the attitude maximization algorithm, RR {\em graphs} need to be stored as opposed to RR {\em sets}.
This leads to high memory usage and also since processing RR graphs is more expensive than processing RR sets, this algorithm is not as scalable as one would like to be.

\section{Experimental Evaluation}
\label{sec:expt}


\begin{wraptable}{l}{0.3\textwidth}
\begin{center}
\scriptsize
\begin{tabular}{|l||r|r|}
\hline 
\textbf{Network-name} & \textbf{\# Nodes} & \textbf{\# Edges} \\ \hline\hline
ego-Facebook & 4039 & 88234 \\ \hline
NetHept & 15229  & 62752 \\ \hline
Epinions & 75888  & 508837 \\ \hline
Amazon  & 334863  & 925872 \\ \hline
DBLP  & 317080  & 1049866 \\ \hline
Youtube & 1134890  & 2987624  \\ \hline
\hline
\end{tabular}
\end{center}
\caption{Datasets}
\label{tbl:datasets}
\vspace{-2em}
\end{wraptable}
Table ~\ref{tbl:datasets} lists the networks used.  The first six
networks are publicly available\footnote{Datasets are obtained from
  \url{http://snap.stanford.edu/data/} and
  \url{https://microsoft.com/en-us/research/people/weic/}. The code is
  available at \url{
    https://github.com/madhavanrp/QuantifyingAttitude}}.


\noindent
\textbf{Experimental Settings.\ } All the algorithms are implemented
in C++ and run on Linux server with AMD Opteron 6320 CPU (8 cores and
2.8 GHz) and 128GB main memory. To estimate the total attitude using
Algorithm ~\ref{algo:totalattitude}, we set
$\epsilon=0.1,\delta=0.001$. As pointed out in ~\cite{Arora:2017}, algorithms that use reverse sampling run into high memory usage owing to the number of samples generated. To find the \attup Maximizing seed set, we use the ideas from the
Stop-and-Stare algorithm ~\cite{NTD16,Huang:2017} that was developed
for the influence maximization problem. This ensures that we generate
(approximately) correct number of RR sets resulting in lesser memory used. It can be proved that this
implementation has the same theoretical guarantees as Algorithm 2.
The source code can be found at \url{https://github.com/madhavanrp/QuantifyingAttitude}.

\newcommand{\budgetVPlot}[2]{

\begin{tikzpicture}
    \begin{axis}[xlabel=\textbf{Budget}, title=\textbf{#2}, legend columns=-1, legend style={at={(0.5,-0.1)},anchor=north,draw=none}]
        \pgfplotstableread[col sep=comma]{datasets/#1}
        \datatablebudget
        \addplot table[y = facebook] from \datatablebudget ;
        \addlegendentry{Facebook}
        
        \addplot table[y = nethept] from \datatablebudget ;
        \addlegendentry{Nethept}
        
        \addplot table[y = amazon] from \datatablebudget ;
        \addlegendentry{Amazon}
        
        \addplot table[y = dblp] from \datatablebudget ;
        \addlegendentry{DBLP}
        
    \end{axis}
\end{tikzpicture}
}

\newcommand{\budgetVPlotYoutube}[2]{

\begin{tikzpicture}
    \begin{axis}[xlabel=\textbf{Budget}, title=\textbf{#2}, legend columns=-1, legend style={at={(0.5,-0.1)},anchor=north,draw=none}]
        \pgfplotstableread[col sep=comma]{datasets/#1}
        \datatablebudget
        
        \addplot table[y = youtube] from \datatablebudget ;
        \addlegendentry{Youtube}
        
    \end{axis}
\end{tikzpicture}}

\begin{figure*}[htp]
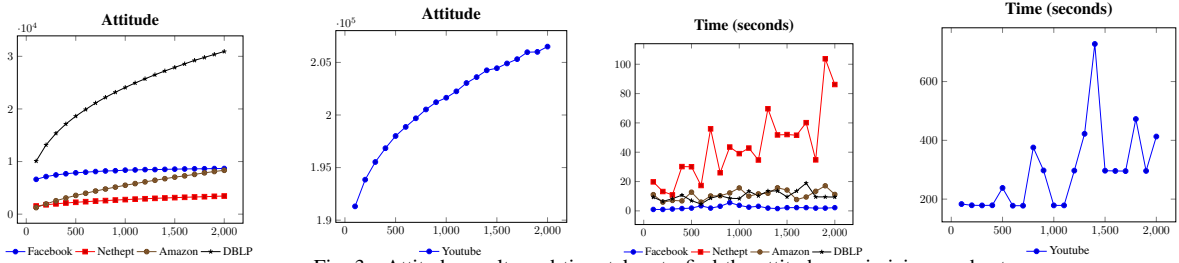

\begin{tabular}{cccc}
\resizebox {0.2\linewidth} {!} {
\budgetVPlot{budgetVsAtt.csv}{\Large Attitude}
}
&
\resizebox {0.2\linewidth} {!} {
\budgetVPlotYoutube{budgetVsAtt.csv}{\Large Attitude}
}
&
\resizebox {0.2\linewidth} {!} {
\budgetVPlot{budgetVsTime.csv}{\Large Time (seconds)}
}
&
\resizebox {0.2\linewidth} {!} {
\budgetVPlotYoutube{budgetVsTime.csv}{\Large Time (seconds)}
}
\end{tabular}
\vspace{-1em}
\caption{\attup results and time taken to find the \att maximizing seed set}
\label{fig:BudgetVs}
\end{figure*}

\smallskip
\noindent
\textbf{Maximizing \attup.\ }
\label{sec:att-inf}
The results are shown in Figure~\ref{fig:BudgetVs} (x-axis represents the
seed set size and the y-axis indicates the attitude or time). The \att
results produced across a wide range of graph sizes demonstrate the
scalability of $RAS$-based maximization.  We computed the attitude
maximization seed set for budgets in the range $[1, 2000]$. As
expected as seed set size increases, the total attitude also
increases. Note that for small networks, the total attitude does not
increase much after certain point. This is due to the submodularity of
the attitude function. After some point, the gain in attitude becomes
minimal.  The time taken to compute the seed set does not increase
much as the seed set size increases. For example, on $DBLP$
($n=317080, m=1049866$), the time taken is less than 20 seconds for
budgets ranging from $100-2000$.  This is due to the fact that as the
seed set size increases, the value of $\sigma(S^*)$ would increase
thus resulting in smaller RR sets (as per the stop-and-stare algorithm).

\graphicspath{ {./images/} }
\begin{figure}\label{ChangeProb}
\begin{centering}
\begin{tabular}{cc}
\resizebox {0.45\linewidth} {0.35\linewidth} {
	\begin{tikzpicture}
	\tikzstyle{every node}=[font=\fontsize{8}{8}\selectfont]
	\begin{axis}[
	ybar,
	enlargelimits=0.15,
	legend style={at={(0.5,-0.2)},
		anchor=north,legend columns=-1},
	ylabel={\attup},
	symbolic x coords={0.02, 0.05, 0.1, 1/indegree},
	xtick=data
	]
	\addplot 
	coordinates {(0.02,182.6666667) (0.05,508.6666667) (0.1,1241) (1/indegree,7425.333333)};
	\addplot 
	coordinates {(0.02,487) (0.05,2194.666667) (0.1,10057) (1/indegree,59773.66667)};
	\legend{Amazon, DBLP}
	\end{axis}
	\end{tikzpicture}
}
&
\resizebox {0.45\linewidth} {0.35\linewidth} {
	\begin{tikzpicture}
	\tikzstyle{every node}=[font=\fontsize{8}{8}\selectfont]
	\begin{axis}[
	ybar,
	enlargelimits=0.15,
	legend style={at={(0.5,-0.2)},
		anchor=north,legend columns=-1},
	ylabel={Time (seconds)},
	symbolic x coords={0.02, 0.05, 0.1, 1/indegree},
	xtick=data
	]
	\addplot 
	coordinates {(0.02,40.05666667) (0.05,22.15666667) (0.1,8.36) (1/indegree,4.713333333)};
	\addplot 
	coordinates {(0.02,20.75666667) (0.05,7.87) (0.1,2.483333333) (1/indegree,1.063333333)};
	\legend{Amazon, DBLP}
	\end{axis}
	\end{tikzpicture} 
}
\end{tabular}
\end{centering}
\vspace{-1em}
\caption{Varying probability with $k=100$}
\label{fig:varyingprob}
\end{figure}
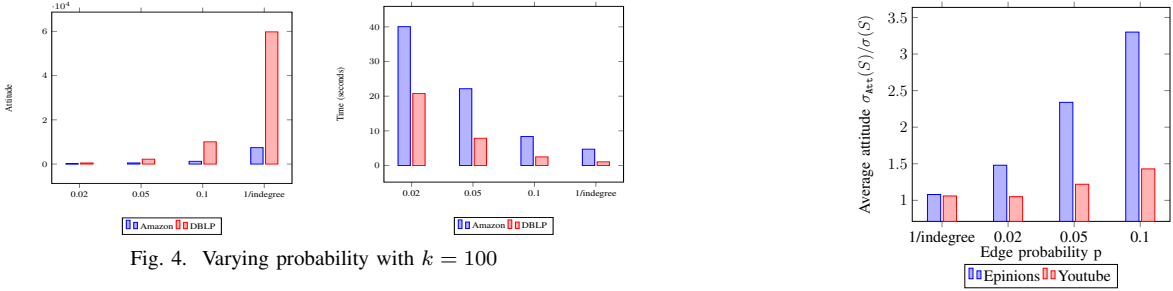

\smallskip
\noindent\textbf{ Propagation Probability and \attup.\ } We consider
different edge probabilities such as $0.02, 0.05, 0.1$ and
$1/\mbox{inDegree}$. The overall \att increases as the probability
increases (See Figure~\ref{fig:varyingprob}). Interestingly, the
maximum \att is observed when the probability is
$1/\mbox{inDegree}$. This is explained by considering the fact that
for each node, it is expected that one of its incoming edges is
activated (if its neighbors are activated). Therefore, the overall
\att is significantly higher if $1/\mbox{inDegree}$ is greater than
$0.1$, on average. We also report how time varies with probability.
We observe that the time taken is least when the edge probability is
$1/\mbox{inDegree}$ and is highest when the probability is
$0.02$. This is again explained by observing that $\sigatt(S^*)$ inversely
impacts the number of RR sets required for estimating \att.  We observe
that this is consistent with the time taken to compute the best seed
with propagation probabilities that produce relatively smaller overall
\att.

\smallskip
\noindent
\textbf{Average Attitude.\ }
Next, we focus on the {\em average attitude of a node}.  There are two
ways to look at this number. The first is the ratio
$\sigatt(S)/\sigma(S)$ which is the ratio of expected attitude and
expected number of influenced nodes. Another measure for average
attitude is to take the expectation of the following ratio: Total
Attitude/Number of nodes influenced. These two quantities need not be
equal, in general, as expectation of a ratio is not the ratio of
expectations. We computed the former quantity by running the presented
algorithms. 
We estimated the latter quantity by running simulations (20000).  The
results are shown in
Table~\ref{tb2:averageAttForAllGraphs}. 
\begin{table}[t]
\scriptsize
	\centering
\begin{tabular}{ | m{0.088\textwidth} | m{0.09\textwidth}| m{0.07\textwidth} | m{0.108\textwidth} | }
\hline
graph name & $\frac{\sigatt(S)}{\sigma(S)}$ & $E[\frac{\texttt{Att}}{\texttt{Inf}}]$ & 
Average indegree \\
			\hline
			ego-Facebook  & 3.21 & 3.20 & 21.85 \\
			\hline
			Epinions  & 3.30 & 3.32 & 6.71 \\
			\hline
			NetHept  & 1.34 & 1.38 & 4.12\\
			\hline
			DBLP  & 1.23 & 1.23 & 3.31 \\
			\hline
			Youtube  & 1.43 & 1.44 & 2.63 \\
			\hline
		\end{tabular}
		 budget = 100 and edge probability = 0.1
	\caption{Average Attitude }
		\label{tb2:averageAttForAllGraphs}	
\vspace{-1em}
\end{table}
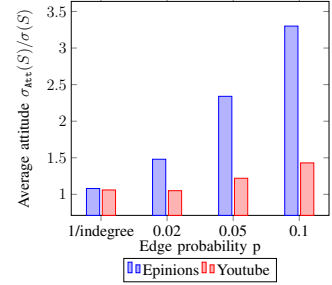
\begin{figure}[h]
\begin{center}
	\begin{tikzpicture}[scale=0.5]
	\tikzstyle{every node}=[font=\fontsize{12}{12}\selectfont]
	\begin{axis}[
	ybar,
	enlargelimits=0.15,
	legend style={at={(0.5,-0.2)},
		anchor=north,legend columns=-1},
	ylabel={Average attitude $\sigatt(S) / \sigma(S)$},
	xlabel = {Edge probability p},
	symbolic x coords={1/indegree, 0.02, 0.05, 0.1},
	xtick=data
	]
	\addplot 
	coordinates {(1/indegree,1.08) (0.02,1.48) (0.05,2.34) (0.1,3.3)};
	\addplot 
	coordinates {(1/indegree,1.06) (0.02,1.05) (0.05,1.22) (0.1,1.43)};
	\legend{Epinions, Youtube}
	\end{axis}
	\end{tikzpicture}
\end{center}
\caption{Average attitude trends as edge probability p increases(k = 100)}
\label{fig:averageAttitudeTrend}
\end{figure}
Interestingly both the
quantities turn out be almost the same for all the graphs. For all the
graphs listed, the average attitudes calculated as $\sigatt(S) /
\sigma(S)$ are greater than 1 as expected since every influenced node
has \att greater than or equal to 1, and they match very well with the
results from the diffusion. 
Graphs with higher average indegrees tend to achieve higher
average attitudes. For example, Epinions achieves a higher average \att
than NetHept.
With increasing edge probabilities, the average attitude
increases(Fig.~\ref{fig:averageAttitudeTrend}) because with
higher edge probabilities, nodes are more likely to be
activated; and with more activated neighbors, a node tends to be
influenced multiple times.

\newcommand{\budgetVEffAttPlot}[2]{
\begin{tikzpicture}
    \begin{axis}[xlabel=\textbf{Budget}, title=\textbf{#2}, legend columns=-1, legend style={at={(0.5,-0.1)},anchor=north,draw=none}]
        \pgfplotstableread[col sep=comma]{datasets/#1}
        \datatablebudget
        \addplot table[y = facebook] from \datatablebudget ;
        \addlegendentry{\Large Facebook}
        
        \addplot table[y = nethept] from \datatablebudget ;
        \addlegendentry{\Large Nethept}
        
        \addplot table[y = dblp] from \datatablebudget ;
        \addlegendentry{\Large DBLP}


    \end{axis}
\end{tikzpicture}
}

\smallskip
\noindent
\textbf{Maximizing \Eff.\ } We implement Algorithm
~\ref{algo:bestSEffectivAttitude} to find the seed set that maximizes
the \Eff. For each $v\in V$, we generate $O(Indegree(v)/\epsilon^2)$
RR graphs where $\epsilon=0.1$. Figure ~\ref{fig:budgetVsEffAtt}
examines the \Eff while varying the budget. We fix the probability to
0.05. As expected, the \Eff does increase when the seed set size is
increased. We observe that the \Eff grows in larger quantities for
$Facebook$ than for the other graphs. This is due to the fact that
$Facebook$ is denser, leading to a higher number of edges activated by
the seed set. We also study how the \attup Maximizing seed compares
with the \Eff Maximizing seed.
%
\begin{wraptable}{l}{0.25\textwidth}
\scriptsize
	\centering
		\begin{tabular}{ | m{0.088\textwidth} | l | l| }
			\hline
			Graph & Alg. ~\ref{algo:greedyapprox} & Alg. ~\ref{algo:bestSEffectivAttitude} \\
			\hline
			ego-Facebook  & 2.11 & 2.69  \\
			\hline
			NetHept  & 1.24 & 1.34\\
            \hline
            Amazon  & 1.01 & 1.03\\
			\hline
			DBLP  & 1.18 & 2.32 \\
			\hline
		\end{tabular}
		 \caption{$E[Att / Inf]$ values for $k=100, p= 0.05$}
	\label{tb3:averageAttComparison}		
\vspace{-1em}
\end{wraptable}
Across various graphs, we note that the \Eff Maximizing seed set
activates fewer nodes when compared to the \attup Maximizing seed. For
example, on $DBLP$ with $k=100,p=0.05$, Attitude maximization
algorithm produces \attup of $2294$ with influence $1930$. In the same
setting, the actionable attitude maximization algorithm produces \attup of
$870$ with influence $376$.  We note two points. The objective
function $\ExpEff(.)$ is higher for the seed set produced by the
actionable attitude maximization compared to the seed set produced by
the attitude maximization problem. Very interestingly, for the
attitude maximization seed set the average attitude is $2294/1930$
which is $1.19$ whereas the actionable attitude maximization seed
results in an average attitude of $870/376$ which is $2.31$.  Recall
that the notion of actionable attitude attempts to maximize entities
that are strongly influenced and thus should result in higher average
attitude and the experiments concur with this intuition. Table
~\ref{tb3:averageAttComparison} compares average attitude for the seed
sets produced by the attitude maximization and actionable attitude
maximization algorithms. The Average Attitude tends to be
  higher when the \Eff is maximized with $Amazon$ being an outlier.

\begin{figure}[h]

\begin{center}
\resizebox {0.2\textwidth} {!} {
\budgetVEffAttPlot{budgetVsAttMinusInfluence.csv}{\Large \Eff}

}
\end{center}
\caption{Budget Vs \Eff, $p=0.05$}
\label{fig:budgetVsEffAtt} 
\end{figure}
These observations suggest that \Eff maximization produces fewer overall nodes activated but with higher individual \attup. As with maximizing \attup, we compared our implementation with the same baseline heuristics observed higher \Eff. The experiments on $Youtube$ do not finish as the program runs out of memory. This is due to the fact that \Eff Maximizing requires the \emph{RR Graphs} to be stored rather than just vertices.



%

\smallskip
\noindent
\textbf{Attitude Distribution.\ } We consider distribution of nodes
with certain attitude values and their contribution to the total
attitude. 
For each attitude value $a$, we looked at the total contribution of
all nodes with attitude $a$ (obtained by multiplying number of nodes
with attitude $a$). The attitude values are on $x$-axis and the
attitude contribution on $y$-axis of
Figure~\ref{fig:averageHistograms}.  
\begin{wrapfigure}{l}{0.25\textwidth}
	\begin{tikzpicture}[scale=0.5]
	\tikzstyle{every node}=[font=\fontsize{10}{10}\selectfont]
	\begin{axis}[
	ybar, 
	bar width=5,
	enlargelimits=0.06, legend style={at={(0.5,-0.2)},
		anchor=north,legend columns=-1},
	ylabel={\Large Average attitude contribution},
	xlabel = {\Large Attitude},
	symbolic x coords={1, 2, 3, 4, 5, 6, 7, 8, 9, 10, 11, 12, 13, 14, 15, 16, 17,18, 19, 20, more},
	xtick=data,
	tick label style={font=\tiny}
	]
	\addplot 
	coordinates {(1, 5905.07) (2, 3038.62) (3,2308.026) (4,1940.66) (5,1689.95) (6,1496.262) (7, 1336.293)(8, 1201.6)
		(9, 1081.899) (10, 976.354) (11, 881.034) (12, 796.2036) (13, 720.1493) (14,656.1086) (15, 598.332) (16, 549.3696) (17, 507.2613)
		(18, 472.509) (19, 444.6817) (20,420.56) (more, 7836.0569)};
	
	\end{axis}
	\end{tikzpicture}
\vspace{-1em}
	\caption{Attitude contributions}
	\label{fig:averageHistograms}
\end{wrapfigure}
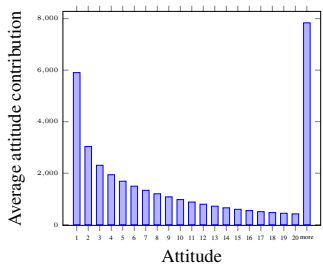
On Epinions graph (with budget
$100$ and edge probability $0.1$) the total expected attitude is
around $34000$ and the expected number of influenced nodes is around
$10,500$.  However, there are 233 nodes whose attitude is more than
$20$ (last bar in the figure). These nodes alone contribute $8,000$ to
the total attitude.  Thus $2\%$ of the influenced nodes contribute
nearly $23\%$ to the total attitude.  This means a relatively small
fraction of nodes with high \att contribute significantly to total
attitude and thus average attitude.

\begin{figure}
\begin{centering}
\begin{tabular}{cc}
\resizebox {0.4\linewidth} {0.35\linewidth} {
\includegraphics{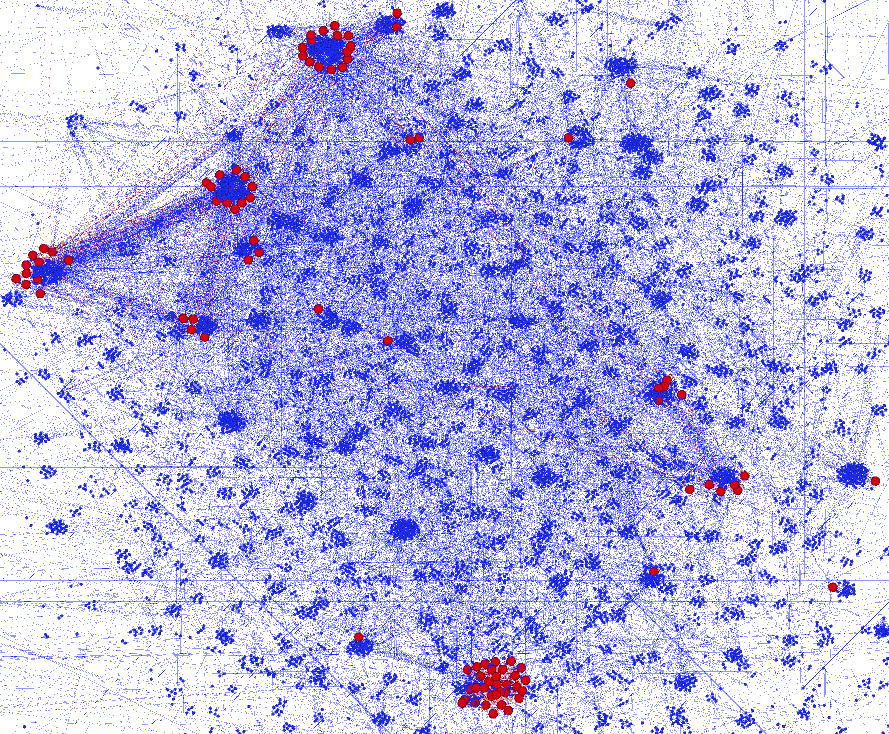}
}
&
\resizebox {0.4\linewidth} {0.35\linewidth} {
\includegraphics{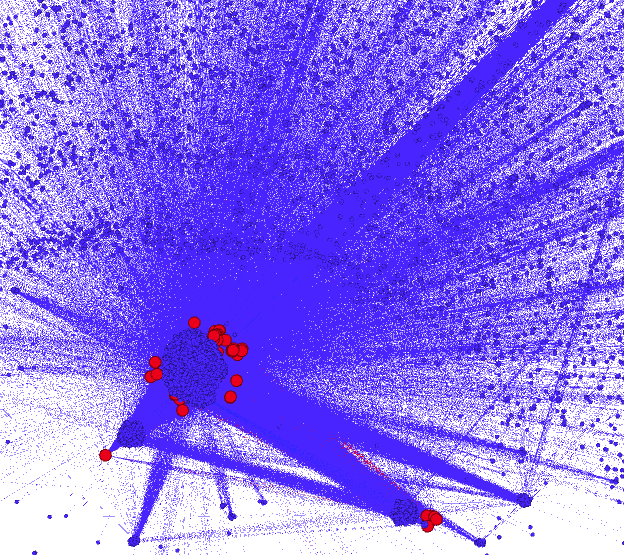}
}
\end{tabular}
\end{centering}
\caption{Clusters of  High \attup nodes}
\label{fig:spatial}
\end{figure}

\smallskip
\noindent
\textbf{Spatial Proximity of Nodes with High \att.\ } Finally we
visualized the location of nodes with high attitude values
(Figure~\ref{fig:spatial}).  Red nodes are the nodes with high
attitude. We used the clustering algorithm mentioned in
~\cite{Blondel_2008} to identify communities, and visualized them
using the OpenOrd algorithm ~\cite{openord} from Gephi
~\cite{ICWSM09154} which is used for visually distinguishing
clusters. For graph Epinions, a total of 708 communities were
identified. We we looked at the top 100 \att nodes, we noticed that
all these nodes were limited to only 5 of those
communities. Similarly, for graph CA-HepTh, 473 communities were
identified. The top 100 \att nodes were limited to 12 of them. This
behavior was observed in other graphs as well, which showed that high
\att nodes are generally restricted to a few communities rather than
being distributed across the network.


\section{Conclusion}
\label{sec:conc}

In this work we have formalized the notion of {\em strength of
  influence/attitude} in social networks and have formulated the
attitude maximization problem. We present various theoretical
properties related to our formulation. We also introduce the notion of
actionable attitude to capture high attitude nodes by defining this
quality as the (expected) difference between total attitude and the
number of influenced nodes.  There are several other ways to formulate
this notion---for example, by looking at the ratio of attitude and
influence or by examining the number of entities whose attitude value
is above a threshold.  Exploring these alternative formulations would be interesting.


\bibliographystyle{plain}
\bibliography{main}

\begin{thebibliography}{10}

\bibitem{Aggarwal:SDM11}
Charu~C. Aggarwal, Arijit Khan, and Xifeng Yan.
\newblock On flow authority discovery in social networks.
\newblock In {\em Proceedings of SIAM International Conference on Data Mining},
  pages 522--533, 2011.

\bibitem{Ajzen01}
I.~Ajzen.
\newblock Nature and operation of attitudes.
\newblock {\em Annual Review of Psychology}, 52:27--58, 2001.

\bibitem{Arora:2017}
A.~Arora, S.~Galhotra, and S.~Ranu.
\newblock Debunking the myths of influence maximization: An in-depth
  benchmarking study.
\newblock In {\em {SIGMOD}}, pages 651--666, 2017.

\bibitem{Barbieri:icdm12}
N.~Barbieri, F.~Bonchi, and G.~Manco.
\newblock Topic-aware social influence propagation models.
\newblock In {\em {ICDM}}, pages 81--90, 2012.

\bibitem{ICWSM09154}
Mathieu Bastian, Sebastien Heymann, and Mathieu Jacomy.
\newblock Gephi: An open source software for exploring and manipulating
  networks.
\newblock 2009.

\bibitem{Blondel_2008}
Vincent~D Blondel, Jean-Loup Guillaume, Renaud Lambiotte, and Etienne Lefebvre.
\newblock Fast unfolding of communities in large networks.
\newblock {\em Journal of Statistical Mechanics: Theory and Experiment},
  2008(10):P10008, 2008.

\bibitem{Borgs14}
C.~Borgs, M.~Brautbar, J.~Chayes, and B.~Lucier.
\newblock Maximizing social influence in nearly optimal time.
\newblock In {\em {SODA}}, pages 946--957, 2014.

\bibitem{Chen:2015}
S.~Chen, J.~Fan, G.~Li, J.~Feng, K-L. Tan, and J.~Tang.
\newblock Online topic-aware influence maximization.
\newblock {\em Proc. VLDB Endow.}, 8(6):666--677, 2015.

\bibitem{negative11}
W.~Chen, A.~Collins, R.~Cummings, T.~Ke, Z.~Liu, D.~Rinc{\'{o}}n, X.~Sun,
  Y.~Wang, W.~Wei, and Y.~Yuan.
\newblock Influence maximization in social networks when negative opinions may
  emerge and propagate.
\newblock In {\em {SDM}}, pages 379--390, 2011.

\bibitem{chen:kdd10}
W.~Chen, C.~Wang, and Y.~Wang.
\newblock Scalable influence maximization for prevalent viral marketing in
  large-scale social networks.
\newblock In {\em {SIGKDD}}, pages 1029--1038, 2010.

\bibitem{domingos:kdd01}
P.~Domingos and M.~Richardson.
\newblock Mining the network value of customers.
\newblock In {\em KDD}, pages 57--56, 2001.

\bibitem{AlwinScott96}
A.~Duane and J.~Scott.
\newblock {\em Understanding Change in Social Attitudes}.
\newblock 1996.

\bibitem{FishbeinAjzen75}
M.~Fishbein and I.~Ajzen.
\newblock {\em Belief, Attitude, Intention, and Behavior: An Introduction to
  Theory and Research}.
\newblock Addison-Wesley, 1975.

\bibitem{GalhotraAroraRoy16}
S.~Galhotra, A.~Arora, and S.~Roy.
\newblock Holistic influence maximization: Combining scalability and efficiency
  with opinion-aware models.
\newblock In {\em {SIGMOD}}, pages 743--758, 2016.

\bibitem{Guo:cikm13}
J.~Guo, P.~Zhang, C.~Zhou, Y.~Cao, and L.~Guo.
\newblock Personalized influence maximization on social networks.
\newblock In {\em CIKM 13}, pages 199--208, 2013.

\bibitem{Huang:2017}
Keke Huang, Sibo Wang, Glenn Bevilacqua, Xiaokui Xiao, and Laks V.~S.
  Lakshmanan.
\newblock Revisiting the stop-and-stare algorithms for influence maximization.
\newblock {\em Proc. VLDB Endow.}, 10(9):913--924, May 2017.

\bibitem{jung:icdm12}
K.~Jung, W.~Heo, and W.~Chen.
\newblock {IRIE:} scalable and robust influence maximization in social
  networks.
\newblock In {\em {ICDM} 2012.}, pages 918--923, 2012.

\bibitem{kempe:kdd03}
D.~Kempe, J.~Kleinberg, and E.~Tardos.
\newblock Maximizing the spread of influence through a social network.
\newblock In {\em KDD}, pages 137--146, 2003.

\bibitem{JMLR:v9:krause08a}
Andreas Krause, Ajit Singh, and Carlos Guestrin.
\newblock Near-optimal sensor placements in gaussian processes: Theory,
  efficient algorithms and empirical studies.
\newblock {\em Journal of Machine Learning Research}, 9(8):235--284, 2008.

\bibitem{leskovec:kdd07}
J.~Leskovec, A.~Krause, C.~Guestrin, C.~Faloutsos, J.~VanBriesen, and
  N.~Glance.
\newblock Cost-effective outbreak detection in networks.
\newblock In {\em KDD}, pages 420--429, 2007.

\bibitem{li:privacy11}
F.H. Li, C.T. Li, and M.K. Shan.
\newblock Labeled influence maximization in social networks for target
  marketing.
\newblock In {\em PASSAT/SocialCom 2011}, pages 560--563, 2011.

\bibitem{LiZhangTan15}
Y.~Li, D.~Zhang, and K-L. Tan.
\newblock Real-time targeted influence maximization for online advertisements.
\newblock {\em VLDB}, 8(10):1070--1081, 2015.

\bibitem{Liu:ICDM13}
Q.~{Liu}, B.~{Xiang}, L.~{Zhang}, E.~{Chen}, C.~{Tan}, and J.~{Chen}.
\newblock Linear computation for independent social influence.
\newblock In {\em IEEE 13th International Conference on Data Mining}, pages
  468--477, 2013.

\bibitem{Liu:TKDD17}
Qi~Liu, Biao Xiang, Nicholas~Jing Yuan, Enhong Chen, Hui Xiong, Yi~Zheng, and
  Yu~Yang.
\newblock An influence propagation view of pagerank.
\newblock {\em ACM Transaction of Knowledge Discovery Data}, 11(3), 2017.

\bibitem{LuL12}
W.~Lu and L.~V.~S. Lakshmanan.
\newblock Profit maximization over social networks.
\newblock In {\em {ICDM}}, pages 479--488, 2012.

\bibitem{openord}
Shawn Martin, W~Michael~Brown, Richard Klavans, and Kevin Boyack.
\newblock Openord: An open-source toolbox for large graph layout.
\newblock {\em Proc SPIE}, 7868:786806, 01 2011.

\bibitem{nemhauser}
George Nemhauser, Laurence Wolsey, and M~L.~Fisher.
\newblock An analysis of approximations for maximizing submodular set
  functions.
\newblock 14:265--294, 12 1978.

\bibitem{NTD16}
H.~Nguyen, M.~Thai, and T.~Dinh.
\newblock Stop-and-stare: Optimal sampling algorithms for viral marketing in
  billion-scale networks.
\newblock In {\em Proceedings {SIGMOD}}, pages 695--710, 2016.

\bibitem{PSBP18}
M.~Padmanabhan, N.~Somisetty, S.~Basu, and A.~Pavan.
\newblock Influence maximization in social networks with non-target
  constraints.
\newblock In {\em {IEEE} International Conference on Big Data, Big Data 2018},
  pages 771--780, 2018.

\bibitem{Rokeach70}
M.~Rokeach.
\newblock {\em Beliefs, Attitudes and Values}.
\newblock Jossey-Bass, 1970.

\bibitem{Song:cikm16}
C.~Song, W.~Hsu, and M.~L. Lee.
\newblock Targeted influence maximization in social networks.
\newblock In {\em Proc. of CIKM 16}, pages 1683--1692, 2016.

\bibitem{Tang15}
Y.~Tang, Y.~Shi, and X.~Xiao.
\newblock Influence maximization in near-linear time: {A} martingale approach.
\newblock In {\em {SIGMOD}}, pages 1539--1554, 2015.

\bibitem{Tang14}
Y.~Tang, X.~Xiao, and Y.~Shi.
\newblock Influence maximization: near-optimal time complexity meets practical
  efficiency.
\newblock In {\em {SIGMOD}}, pages 75--86, 2014.

\bibitem{Zajonc68}
R.~Zajonc.
\newblock Attitudinal effects of mere exposure.
\newblock {\em Journal of Personality and Social Psychology.}, 9(2):1--27,
  1968.

\bibitem{ZhangDinhThai13}
H.~Zhang, T.~N. Dinh, and M.~T. Thai.
\newblock Maximizing the spread of positive influence in online social
  networks.
\newblock In {\em {ICDCS}}, pages 317--326, 2013.

\bibitem{Zhou:ICCS14}
Chuan Zhou, Peng Zhang, Wenyu Zang, and Li~Guo.
\newblock Maximizing the cumulative influence through a social network when
  repeat activation exists.
\newblock In {\em International Conference on Computational Science}, pages 422
  -- 431, 2014.

\end{thebibliography}

\end{document}